\documentclass{article}
\usepackage{PRIMEarxiv}

\usepackage[hidelinks]{hyperref}
\usepackage{booktabs}
\usepackage{amsfonts}
\usepackage{float}
\usepackage{bbm}
\usepackage{algorithm}
\usepackage{algpseudocode}
\usepackage{adjustbox}
\usepackage{amsmath,amsthm,amssymb,amsfonts, fancyhdr, color, comment, graphicx, environ}
\usepackage[citestyle=authoryear,natbib=true,backend=bibtex,doi=false,isbn=false,url=false,eprint=false]{biblatex}

\usepackage{mathtools}
\newtheorem{theorem}{Theorem}[section]
\newtheorem{lemma}[theorem]{Lemma}

\DeclareMathOperator*{\argmax}{arg\,max}
\DeclareMathOperator*{\argmin}{arg\,min}

\DeclarePairedDelimiter{\ceil}{\lceil}{\rceil}
\DeclarePairedDelimiter\floor{\lfloor}{\rfloor}

\title{Conformal Contextual Robust Optimization}

\author{
  Yash Patel, Sahana Rayan, Ambuj Tewari \\
  Department of Statistics \\
  University of Michigan \\
  \texttt{\{yppatel,srayan,tewaria\}@umich.edu} \\
}

\author{
  Yash Patel \\
  Department of Statistics \\
  University of Michigan \\
  \texttt{yppatel@umich.edu} \\
   \And
  Sahana Rayan \\
  Department of Statistics \\
  University of Michigan \\
  \texttt{srayan@umich.edu} \\
   \And
  Ambuj Tewari \\
  Department of Statistics \\
  University of Michigan \\
  \texttt{tewaria@umich.edu} \\
}

\begin{document}
\maketitle

%

%





\begin{abstract}
  Data-driven approaches to predict-then-optimize decision-making problems seek to mitigate the risk of uncertainty region misspecification in safety-critical settings. Current approaches, however, suffer from considering overly conservative uncertainty regions, often resulting in suboptimal decision-making. To this end, we propose Conformal-Predict-Then-Optimize (CPO), a framework for leveraging highly informative, nonconvex conformal prediction regions over high-dimensional spaces based on conditional generative models, which have the desired distribution-free coverage guarantees. Despite guaranteeing robustness, such black-box optimization procedures alone inspire little confidence owing to the lack of explanation of why a particular decision was found to be optimal. We, therefore, augment CPO to additionally provide semantically meaningful visual summaries of the uncertainty regions to give qualitative intuition for the optimal decision.
We highlight the CPO framework by demonstrating results on a suite of simulation-based inference benchmark tasks and a vehicle routing task based on probabilistic weather prediction. 
\end{abstract}

\section{INTRODUCTION}\label{section:intro}
Predict-then-optimize or contextual robust optimization problems are of long-standing interest in safety-critical settings where decision-making happens under uncertainty \citep{sun2023predict, elmachtoub2022smart, elmachtoub2020decision, pervsak2023contextual}. In traditional robust optimization, results are made to be robust to distributions anticipated to be present upon deployment \citep{ben2009robust, beyer2007robust}. Since such decisions are sensitive to proper model specification, recent efforts have sought to supplant this with data-driven uncertainty regions \citep{cheramin2021data, bertsimas2018data, shang2019data, johnstone2021conformal}. 

Model misspecification is ever more present in \textit{contextual} robust optimization, spurring efforts to define similar data-driven uncertainty regions \citep{ohmori2021predictive, chenreddy2022data, sun2023predict}. Such methods, however, focus on box- and ellipsoid-based uncertainty regions, both of which are necessarily convex and often overly conservative, resulting in suboptimal decision-making. 

Conformal prediction provides a principled framework for producing distribution-free prediction regions with marginal frequentist coverage guarantees \citep{angelopoulos2021gentle, shafer2008tutorial}. By using conformal prediction on a user-defined score function $s(x,y)$ and obtaining an empirical $1-\alpha$ quantile $\widehat{q}(\alpha)$ of $s(x,y)$ over a calibration set $\mathcal{D}_{\mathcal{C}}$, prediction regions $\mathcal{C}(x) = \{y \mid s(x, y) \le \widehat{q}(\alpha)\}$ attain marginal coverage guarantees. Such prediction regions, however, are notably defined \textit{implicitly}. For simple scores, such as residuals, an explicit expression of such regions can be written, making these the most common approaches used in practice \citep{tumu2023physics,horwitz2022conffusion,angelopoulos2022image,hu2022robust,mao2022valid}. 

\begin{figure*}
  \centering 
  \includegraphics[scale=0.34]{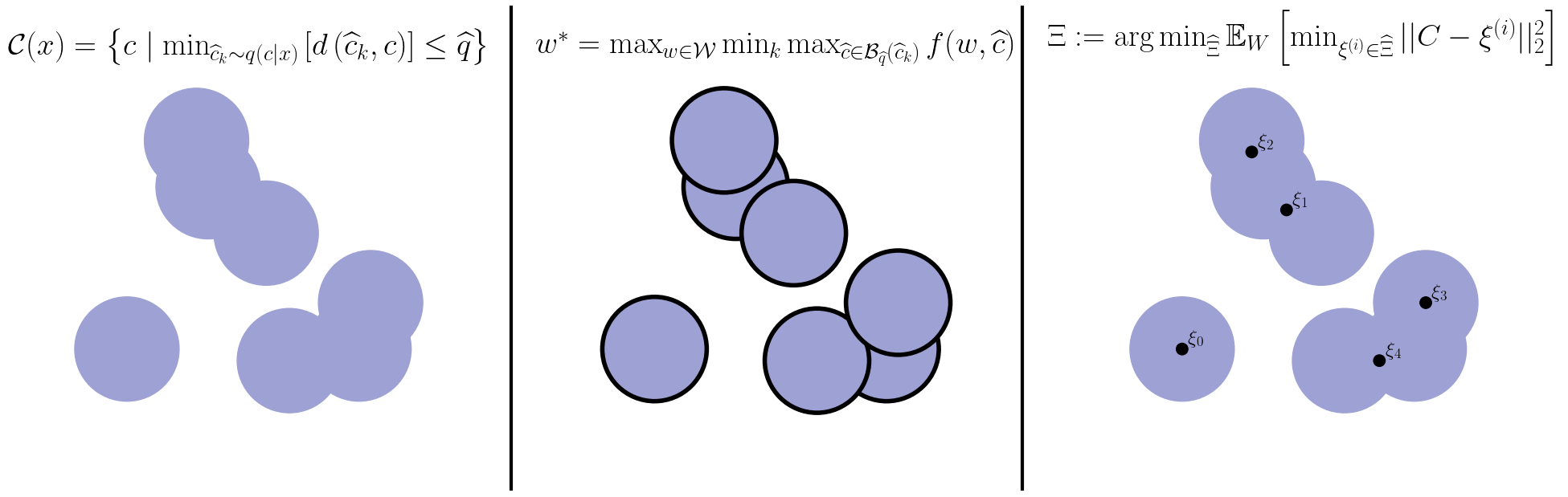}
  \caption{\label{fig:workflow} \textproc{CPO} leverages informative, non-convex conformal prediction regions for robust predict-then-optimize decision making. \textproc{CPO} uses a score function such that the resulting prediction regions can be decomposed into convex subregions over which optimization can be carried out efficiently. Visual summaries $\{\xi^{(i)}\}$ of the prediction region can similarly be efficiently sampled to gain intuition on the optimal decision $w^{*}$.
  }
\end{figure*}

The disadvantage is that such score functions ignore the structure that is often present in high-dimensional data, such as images. 
Choices of simplistic scores, thus, tend to be overly conservative and often produce convex prediction regions even when $\mathcal{P}(Y|X)$ is non-convex. 
Recent work has demonstrated that defining scores using conditional generative models produces sharper and, hence, more informative prediction regions \citep{feldman2023calibrated,wang2022probabilistic,patel2023variational}. We, thus, extend the line of data-driven predict-then-optimize work by considering such generative model-based prediction regions.


In addition to contributing to the predict-then-optimize line of inquiry, we view this work as addressing a concern of the conformal prediction community: how to use implicitly defined non-convex, high-dimensional prediction regions. Works producing such regions have themselves noted the difficulty in their use \citep{sesia2021conformal,izbicki2022cd}. Initial works on coverage for images have framed the utility of their results in highlighting regions of the image with the greatest variability and, hence, uncertainty \citep{angelopoulos2022image,horwitz2022conffusion,belhasin2023principal}. 

Extending such visualization 
gives invaluable intuition to the end user.
For instance, a black-box optimization procedure for producing drug candidates to robustly bind to a predicted protein structure offers little insight into the decision-making process; however, semantic summaries of the uncertainty region would reveal regions of flexibility of the protein, clarifying why particular structures were deemed optimal in the candidate drug. Such interest in explainable robust decision-making was highlighted in a recent survey \citep{sadana2023survey}, especially given the ``right to explanation'' mandated by the EU's ``General Data Protection Regulation'' \citep{doshi2017towards, kaminski2019right}. 
Our main contributions, thus, are:
\begin{itemize}
    \item Proposing Conformal-Predict-Then-Optimize (CPO) to leverage informative, non-convex prediction regions for decision-making.
    \item Providing interpretable visual summaries of uncertainty regions using representative points.
    \item Demonstrating the generality of CPO across a suite of benchmark tasks and a traffic routing task based on probabilistic weather prediction. 
\end{itemize}

\section{BACKGROUND}\label{section:background}

\subsection{Conformal Prediction}
Given a dataset $\mathcal{D}=\{(x^{(1)},y^{(1)}),\ldots (x^{(N)},y^{(N)})\}$ of i.i.d.~observations from a distribution $\mathcal{P}(X, Y)$, conformal prediction \citep{angelopoulos2021gentle, shafer2008tutorial} produces prediction regions with distribution-free theoretical guarantees. A prediction region maps from observations of $X$ to sets of possible values for $Y$ and is said to be marginally valid at the $1-\alpha$ level if $\mathcal{P}_{X, Y}(Y \notin \mathcal{C}(X))\leq \alpha$.

Split conformal is one popular version of conformal prediction. In this approach, marginally calibrated regions $\mathcal{C}$ are designed using a ``score function'' $s(x,y)$. Intuitively, the score function should have the quality that $s(x,y)$ is smaller when it is more reasonable to guess that $Y=y$ given the observation $X=x$. For example, if one has access to a function $\hat f(x)$ which attempts to predict $Y$ from $X$, one might take $s(x,y)=\Vert\hat f(x)-y\Vert$. The score function is evaluated on each point of a dataset $\mathcal{D_C}$ called the ``calibration dataset,'' yielding $\mathcal{S} = \{s(x^{(j)}, y^{(j)})\}_{j=1}^{N_{\mathcal{C}}}$, where $N_{\mathcal{C}} := |\mathcal{D_C}|$. Note that the calibration dataset cannot be used to pick the score function; if data is used to design the score function, it must independent of  $\mathcal{D_C}$. 
We then define $\widehat q(\alpha)$ as the $\ceil{(N_{\mathcal{C}}+1)(1-\alpha)}/N_{\mathcal{C}}$ quantile of $\mathcal{S}$. For any future $x$, the set $\mathcal{C}(x) = \{y \mid s(x, y) \le \widehat{q}(\alpha)\}$ satisfies $1 - \alpha\le\mathcal{P}(Y \in \mathcal{C}(X))$. This inequality is known as the coverage guarantee, and it arises from the exchangeability of the score of a test point $s(x', y')$ with $\mathcal{S}$. The coverage guarantee possesses finite-sample properties.

 As noted in Vovk's tutorial \citep{shafer2008tutorial}, while the coverage guarantee holds for any score function, different score functions may lead to more or less informative prediction regions. For example, the score $s(x,y)=1$ leads to the highly uninformative prediction region of all possible values of $Y$. Predictive efficiency is one way to quantify informativeness, defined as the inverse of the expected Lebesgue measure of the prediction region, i.e. $\left(\mathbb{E}[|\mathcal{C}(X)|]\right)^{-1}$ \citep{yang2021finite,sesia2020comparison}. Methods employing conformal prediction often seek to identify prediction regions that are efficient and calibrated.

\subsection{Representative Points}
The problem of summarizing the distribution of a random vector with points $\Xi := \{\xi^{(i)}\}_{i=1}^N$ arises in many contexts, such as in optimal stratification \citep{dalenius1950problem,dalenius1951problem}, density estimation \citep{flury1993representing}, and signal quantization \citep{max1960quantizing}. Such points are known as representative points (RPs). Denoting the space of all sets $\widehat{\Xi}$ such that $|\widehat{\Xi}| \le n$ as $\zeta$, the RPs of a random variable $X$ are
\begin{equation}
    \Xi := \argmin_{\widehat{\Xi}\in\zeta} \mathbb{E}_{X}\left[\min_{\xi^{(i)} \in \widehat{\Xi}} || X - \xi^{(i)} ||_2^2 \right].
\end{equation}
For a comprehensive review of representative points, see \citep{fang2023review}. Despite extensive study, no general algorithm exists for the efficient construction of representative points for arbitrary distributions. Typical implementations use clustering algorithms, such as Lloyd's algorithm, on $\{x^{(i)}\}_{i=1}^M \sim\mathcal{P}(X)$.

\subsection{Predict-Then-Optimize}
Predict-then-optimize problems are formulated as
\begin{equation}
\begin{aligned}
w^{*}(x) := \min_{w\in\mathcal{W}} \quad & \mathbb{E}[C^{T} w\mid x],
\end{aligned}
\end{equation}
where $w$ are decision variables, $C$ an \textit{unknown} cost parameter, $x$ observed contextual variables, and $\mathcal{W}$ a compact feasible region. The predict-then-optimize framework is so called as the nominal approach first predicts $\widehat{c} := f(x)$ and subsequently solves $\min_{w} \widehat{c}^T w$. 
Alternatively, a predictive contextual distribution $\mathcal{P}(C\mid x)$ is assumed, with respect to which the optimization formulation is solved. A full review is presented in \citep{elmachtoub2022smart}.

This formulation, however, is inappropriate in risk-sensitive downstream tasks. For this reason, recent works have begun investigating a risk-sensitive variant or ``robust'' alternative to this traditional formulation, namely by replacing $\mathbb{E}[C^{T} w\mid x]$ with $\max_{\widehat{c}\in\mathcal{U}(x)} \widehat{c}^{T} w$ \citep{ohmori2021predictive, chenreddy2022data, sun2023predict}, where $\mathcal{U}(x)$ is constructed to guarantee coverage of $c$, precisely stated in Lemma \ref{lemma:coverage_bound}.




\section{METHOD}\label{section:method}
We now propose CPO, a way to perform robust predict-then-optimize decision-making over informative, non-convex prediction regions based on generative models. We then discuss how to construct visual summaries of the contents of the conformal prediction regions using a collection of $N$ representative points. 

\subsection{CPO: Problem Formulation}

Let $c\in\mathcal{C}$, where $(\mathcal{C}, d)$ is a general metric space, and $\mathcal{F}$ be the $\sigma$-field of $\mathcal{C}$. While the standard predict-then-optimize framework assumes a linear objective function $c^{T} w$, we consider general convex-concave objective functions $f(w, c)$ that are $L$-Lipschitz in $c$ under the metric $d$ for any fixed $w$, as follows:
\begin{equation}\label{eqn:robust_po}
\begin{gathered}
w^{*}(x) := \min_{w,\mathcal{U}} \max_{\widehat{c}\in\mathcal{U}(x)} \quad f(w, \widehat{c}) \\
\textrm{s.t.} \quad \mathcal{P}_{X,C}(C\in\mathcal{U}(X)) \ge 1-\alpha, \\
\end{gathered}
\end{equation}
where $\mathcal{U} : \mathcal{X}\rightarrow\mathcal{F}$ is a uncertainty region predictor. Exact solution of this problem is intractable, as no practical methods exist to optimize over the predictor function space $\mathcal{U}$. Practical solution of this optimization problem, thus, involves optimizing over several prespecified uncertainty region predictors $\{\mathcal{U}_i\}_{i=1}^{N}$. For any fixed $\mathcal{U}$, this robust counterpart to the nominal predict-then-optimize problem produces a valid upper bound if $c\in\mathcal{U}(x)$. Denoting the pessimism gap as $\Delta(x, c) := \min_{w} \max_{\widehat{c}\in\mathcal{U}(x)} f(w, \widehat{c}) - \min_{w} f(w, c)$, we clearly see $\Delta(x, c) \ge 0$ if $c\in\mathcal{U}(x)$,
formalized below. 
\begin{lemma}\label{lemma:coverage_bound}
    Consider any $f(w, c)$ that is $L$-Lipschitz in $c$ under the metric $d$ for any fixed $w$. Assume further that $\mathcal{P}_{X,C}(C \in \mathcal{U}(X)) \ge 1 - \alpha$. Then, 
    \begin{equation}
        \mathcal{P}_{X,C}\left(0\le \Delta(X, C) \le L \mathrm{\ diam}(\mathcal{U}(X))\right) \ge 1 - \alpha.
    \end{equation}
\end{lemma}
The proof is deferred to Appendix \ref{section:coverage_bound}. Thus, $1-\alpha$ validity of $\mathcal{U}$ ensures the RO procedure produces a valid bound with probability $1-\alpha$, with more efficient prediction regions resulting in tighter upper bounds.



\subsection{CPO: Score Function}
We assume a conditional generative model $q(C\mid X)$ is learned for this prediction task. 
For most score functions, the min-max optimization problem of Equation \ref{eqn:robust_po} is computationally intractable. Crucially, however, we can consider an extension to the score proposed in \citep{wang2022probabilistic}, which lends itself to a decomposition under which such optimization becomes tractable. For a fixed $K$ and $\{\widehat{c_k}\}_{k=1}^K \sim q(C \mid x)$, let
\begin{equation}\label{eqn:score_gen_pp}
    s(x,c) = \min_{k}\left[d\left(\widehat{c}_k, c\right) \right].
\end{equation}
We refer to this score as ``Generalized Probabilistic Conformal Prediction,'' (GPCP) whose validity follows from that of the original PCP framework \citep{wang2022probabilistic}. We discuss the selection of $K$ in Section \ref{section:k_selection}.

\subsection{CPO: Optimization Algorithm}\label{section:optimization}
We fix $\alpha\in[0,1]$ and take $\mathcal{U}(x)$ to be the $1-\alpha$ prediction region $\mathcal{C}(x)$. 
Let $\phi(w) := \max_{\widehat{c} \in \mathcal{C}(x)} f(w, \widehat{c})$. It follows that $\phi(w)$ is convex by Danskin's Theorem by assumption of the convexity of $f$ in $w$. Exact solution of the min-max problem, thus, follows using standard gradient-based optimization techniques on $\phi(w)$. By Danskin's Theorem, $\nabla_{w} \phi(w) = \nabla_{w} f(w, c^{*})$, where $c^{*} := \max_{\widehat{c} \in \mathcal{C}(x)} f(w, \widehat{c})$. We follow the standard projected gradient descent optimization scheme, projecting into $\mathcal{W}$ at each iterate, denoted by $\Pi_{\mathcal{W}}$.


Efficient solution of this RO problem, therefore, reduces to being able to efficiently solve the maximization problem over $\mathcal{C}(x)$. While challenging over general nonconvex regions, the GPCP score formulation lends itself to a highly structured prediction region, namely of the form $\mathcal{C}(x) = \bigcup_{k=1}^{K} \mathcal{B}_{\widehat{q}}(\widehat{c}_{k})$ with $\mathcal{B}_{\widehat{q}}$ being a ball of radius $\widehat{q}$, the conformal quantile, under the $d$ metric. This decomposition of $\mathcal{C}(x)$ means the maximum can be efficiently computed by aggregating the maxima over the individual balls:
\begin{equation}
    \max_{\widehat{c} \in \mathcal{C}(x)} f(w, \widehat{c})
    = \max_{k} \max_{\widehat{c} \in \mathcal{B}_{\widehat{q}}(\widehat{c}_{k})} f(w, \widehat{c}),
\end{equation}
where the maximum over a ball can be efficiently computed with traditional convex optimization techniques. 
This procedure is summarized in Algorithm \ref{alg:CPO-Opt}. The convergence of this procedure proceeds as follows, whose proof is deferred to Appendix \ref{section:convergence}.

\begin{algorithm}
  \caption{\label{alg:CPO-Opt} \textproc{CPO-Opt}}
  \begin{algorithmic}[1]
    \Procedure{CPO-Opt}{}
    \Statex \textbf{Inputs: } Context $x$, CGM $q(C\mid X)$, Optimization steps $T$, Score samples $K$, Conformal quantile $\widehat{q}$
    \State $w \sim U(\mathcal{W}), \{\widehat{c_k}\}_{k=1}^{K} \sim q(C \mid x)$
    \For{$t \in\{1,\ldots T\}$}
    \State $\left\{c_{k}^{*} \gets \argmax_{\widehat{c} \in \mathcal{B}_{\widehat{q}}(\widehat{c}_{k})} f(w, \widehat{c})\right\}_{k=1}^{K}$
    \State $c^{*} \gets \argmax_{c_{k}^{*}} f(w, c_{k}^{*})$
    \State $w \gets \Pi_{\mathcal{W}} (w - \eta \nabla_{w} f(w, c^{*}))$
    \EndFor
    \State{\textbf{Return} $w$}
    \EndProcedure
    \end{algorithmic}
\end{algorithm}

\begin{lemma}\label{lemma:pgd_convergence_main}
    Let $\phi(w) := \max_{\widehat{c} \in \bigcup_{k=1}^{K} \mathcal{B}_{\widehat{q}}(\widehat{c}_{k})} f(w, \widehat{c})$ for $\{\widehat{c}_{k}\}_{k=1}^{K}\subset\mathcal{C}$, $\widehat{q}\in\mathbb{R}^{+}$, and $f(w, c)$ convex-concave and $L$-Lipschitz in $c$ for any fixed $w$. Let $w^{*}\in \mathcal{W}$ be a minimizer of $\phi$. For any $\epsilon > 0$, define $T := \frac{L^2 || w_0 - w^{*} ||}{\epsilon^2}$ and $\eta := \frac{|| w_0 - w^{*} ||}{L \sqrt{T}}$. Then the iterates $\{w_{t}\}_{t=0}^{T}$ returned by Algorithm \ref{alg:CPO-Opt} satisfy
    \begin{equation}
        \phi\left(\frac{1}{T+1}\sum_{t=0}^{T} w_{t}\right) -  \phi(w^{*}) \le \epsilon.
    \end{equation}
\end{lemma}

\subsection{CPO: $K$ Selection}\label{section:k_selection}
Crucially, the convergence highlighted in Lemma \ref{lemma:pgd_convergence_main} reveals that the number of ``outer'' iterations (i.e. $T$) has no dependence on $K$. This is apparent from the proof, in which the iterate count $T$ hinges upon the Lipschitz constant of $\phi(w) = \max_{k} \max_{\widehat{c} \in \mathcal{B}_{\widehat{q}}(\widehat{c}_{k})} f(w, \widehat{c}) := \max_{k} \phi_{k}(w)$, which critically is $L$-Lipschitz \textit{regardless} of what $K$ is selected, as  each $\phi_{k}(w)$ is $L$-Lipschitz. 

We can, thus, solely focus attention on the impact the choice of $K$ has on the ``inner'' optimization computational cost, namely $\max_{k} \phi_{k}(w)$.
This linearly increasing cost with $K$, however, must be juxtaposed with the improved \textit{statistical} efficiency of such prediction regions. In particular, \citep{wang2022probabilistic} empirically demonstrated region size generally decreased nonlinearly up to a saturation point as a function of $K$. 

Critically, this inflection point can be determined \textit{prior} to performing the optimization, since doing so only requires access to $q(C\mid X)$ and test samples to estimate the prediction region size. As pointed out in \citep{wang2022probabilistic} and proven in \citep{chan2008slightly}, estimation of the volume of a union of hyperspheres is complicated by the need to account for overlapped regions. 
$K$ is, thus, chosen based on Monte Carlo estimates of the prediction region volume using Voronoi cells of the hypersphere centers given by \citep{edelsbrunner1995union}:
\begin{equation}\label{eqn:cp_mc_est}
\widehat{\ell}(\{\mathcal{B}_{\widehat{q}}(\widehat{c}_{k})\}) := |\mathcal{B}_{\widehat{q}}| \sum_{k = 1}^K \mathcal{P}_{C \sim U(\mathcal{B}_{\widehat{q}}(\widehat{c}_k))}(C \in V(\widehat{c}_k)),
\end{equation}
where 
 $C \sim U(\mathcal{B}_{\widehat{q}}(\widehat{c}_k))$ denotes a random variable defined uniformly over the region associated with $\widehat{c}_k$, $|\mathcal{B}_{\widehat{q}}|$ the volume of a hypersphere of radius $\widehat{q}$, and $V(\widehat{c}_k)$ the Voronoi cell of $\widehat{c}_k$, defined as $\{z \in \mathbb{R}^d \mid d(\widehat{c}_k, z) \leq d(\widehat{c}_{k'}, z), k'\neq k\}$. Muller's method enables efficient sampling of $U(\mathcal{B}_{\widehat{q}}(\widehat{c}_k))$ \citep{muller1959note,fishman2013monte}.


We then choose $K^{*}$ to be the inflection point, namely the $\argmin_{K} |\widehat{\ell}_K - \widehat{\ell}_{K + 1}| \leq \epsilon$ for some 
user-specified $\epsilon$ volume tolerance. Critically, these volume estimates must be performed on a distinct subset of the data from $\mathcal{D}_{\mathcal{C}}$ 
as exchangeability with future test points is otherwise lost in conditioning on $\mathcal{D_C}$ for selecting $K^*$ \citep{yang2021finite}. We, thus, partition $\mathcal{D}_{\mathcal{C}} := \mathcal{D}_{\mathcal{C}_1} \cup \mathcal{D}_{\mathcal{C}_2}$, using $\mathcal{D}_{\mathcal{C}_1}$ for calibration and $\mathcal{D}_{\mathcal{C}_2}$ for volume estimation, detailed in Algorithm \ref{alg:CPO}.

\begin{algorithm}
  \caption{\label{alg:CPO} \textproc{CPO}}
  \begin{algorithmic}[1]
    \Procedure{VolumeEst}{}
    \Statex \textbf{Inputs: } Context $x$, CGM $q(C\mid X)$, Conformal quantile $\widehat{q}$
    \State $\{\widehat{c}_k\}_{k = 1}^{K} \sim q(C_{1:K}\mid x)$
    \State $\left\{\{c_{k, m}\}_{m = 1}^{M}  \sim U(\mathcal{B}_{\widehat{q}}(\widehat{c}_{k})) \right\}_{k = 1}^{K}$
    \State \textbf{Return} $|B_{\widehat{q}}| \sum_{k=1}^{K} \frac{1}{M}\sum_{m=1}^{M}\mathbbm{1}\left[c_{k, m} \in V(\widehat{c}_{k})\right]$
    \EndProcedure
    \Statex 
    \Procedure{CPO}{}
    \Statex \textbf{Inputs: } Context $x$, CGM $q(C\mid X)$, Optimization steps $T$, Desired coverage $1-\alpha$, Max samples $K_{\max}$, Volume Tolerance $\epsilon$, Calibration sets $\mathcal{D}_{\mathcal{C}_1}, \mathcal{D}_{\mathcal{C}_2}$
    \For{$K \in\{1,\ldots K_{\max}\}$}
    \State $s_K(x, c) \gets \min_{\widehat{c}_k\in \{\widehat{c}_i\} \sim q(C_{1:K}\mid x) }\left[d\left(\widehat{c}_k, c\right) \right]$
    \State $\mathcal{S}_K \gets \left\{s_K(x^{(i)}, c^{(i)})\mid (x^{(i)}, c^{(i)})\in\mathcal{D}_{\mathcal{C}_1}\right\}$
    \State $\widehat{q}_K \gets \frac{\ceil{(|\mathcal{D}_{\mathcal{C}_1}| +1)(1-\alpha)}}{|\mathcal{D}_{\mathcal{C}_1}|} \text{ quantile of } \mathcal{S}_K$
    \State $\widehat{\ell}_K \gets \frac{1}{|\mathcal{D}_{\mathcal{C}_2}|}\sum_{i=1}^{|\mathcal{D}_{\mathcal{C}_2}|}\textproc{VolumeEst}(x^{(i)}, q, \widehat{q}_K)$
    \EndFor
    \State $K^{*} \gets \argmin_{K} \left|\widehat{\ell}_K - \widehat{\ell}_{K + 1}\right| \leq \epsilon$
    \State{\textbf{Return} $\textproc{CPO-Opt}(x, q, T, K^{*}, \widehat{q}_{K^{*}})$}
    \EndProcedure
    \end{algorithmic}
\end{algorithm}

\subsection{CPO: Representative Points}\label{section:method_rps}
We now frame the problem of summarizing the prediction region $\mathcal{C}(x)$. We critically note that this issue of interpretability is non-existent in traditional approaches to robust predict-then-optimize, where uncertainty regions are interpretable by construction, being balls around nominal estimates $\mathcal{B}_{\epsilon}(\widehat{c})$. In other words, there is a fundamental tension in qualitative interpretability and the expressiveness of uncertainty regions, requiring a bespoke method for recovering intuition when leveraging conformal prediction regions. Formally, for a user-specified number of summary points $N$ and query $x$, we seek
\begin{equation}
    \Xi(x) := \argmin_{\widehat{\Xi}\in\zeta} \mathbb{E}_{C\sim U(\mathcal{C}(x))}\left[\min_{\widehat{\xi}^{(i)} \in \widehat{\Xi}} d(C, \xi^{(i)}) \right].
\end{equation}
We use the shorthand $d(C, \Xi) := \min_{\xi^{(i)} \in \Xi} d(C, \xi^{(i)})$. In other words, we wish to construct representative points for a uniform sampling of the prediction region. A naive approach would simply involve explicitly gridding the output space $\mathcal{C}$, filtering such points with the rejection criterion of $\mathcal{C}(x)$, and clustering the remaining points per the $d$ metric. 
This, however, is intractable in high-dimensional cases. Thus, a sampling method is employed to circumvent gridding, 
paralleling the technique leveraged for volume estimation.

$M$ samples are initially drawn $\{c_{k, m}\}_{m = 1}^{M} \sim U(\mathcal{B}_{\widehat{q}}(\widehat{c}_k))$ for each $k$. Importantly, such uniform sampling of the balls leads to \textit{non}-uniform sampling over $\mathcal{C}(x)$ if naively aggregated across $k$, as overlapped regions will be more densely sampled. For this reason, we subsample by discarding those samples $c_{k, m}$ for which $c_{k, m} \in V(\widehat{c}_{k'})$ for $k\neq k'$. This results in samples $C := \{c_i\}$ drawn from the desired $U(\mathcal{C}(x))$.


RPs must be aggregated separately for each connected subregion of $\Omega_{\ell} \subset \mathcal{C}(x)$ to ensure each $\xi^{(i)} \in \mathcal{C}(x)$. That is, we must identify $C_{\ell} := C \cap \Omega_{\ell}$. To do so, we determine if two points $(c_{i}, c_{j})$ belong to the same $\Omega_{\ell}$ by considering the corresponding connected components problem defined on the graph induced by the edge criterion $e_{i,j} = \mathbbm{1}[d(c_{i}, c_{j}) < \widehat{q}]$.
For each $C_{\ell}$, we find a subset $N_{\ell} :=N(|C_{\ell}| / |C|)$ of the total $N$ representative points, for which we use $\textproc{K-Means++}$ with the $d$ metric. 
The full procedure is in Algorithm \ref{alg:CPO-RPs}.

\begin{algorithm}
  \caption{\label{alg:CPO-RPs} \textproc{CPO-RPs}: $\textproc{QueryBall}(\mathcal{T}, x, r)$ is an assumed subroutine that returns all points in the $k$d tree $\mathcal{T}$ that are within a radius $r$ of $x$.}
  \begin{algorithmic}[1]
    \Procedure{CPO-RPs}{}
    \Statex \textbf{Inputs: } Context $x$, CGM $q(C\mid X)$, RP count $N$, Conformal quantile $\widehat{q}$
    \State $\{\widehat{c}_k\}_{k = 1}^{K} \sim q(C_{1:K}\mid x)$
    \State $\left\{\{c_{k, m}\}_{m = 1}^{M}  \sim U(\mathcal{B}_{\widehat{q}}(\widehat{c}_{k})) \right\}_{k = 1}^{K}$
    \State $C \gets \{c_{k, m} \mid c_{k, m} \in V(\widehat{c}_{k}) \}_{k=1,m=1}^{K,M}$
    \State $\mathcal{T} \gets \textproc{KD-Tree}(C)$
    \State $\mathcal{E} \gets \bigcup_i \{c_{i} \times \textproc{QueryBall}(\mathcal{T}, c_{i}, \widehat{q}) \mid c_{i} \in\mathcal{T} \}$
    \State $\{C_{\ell}\} \gets \textproc{ConnectedComponents}(\mathcal{G}(C, \mathcal{E}))$
    \State $\Xi \gets \bigcup_{\ell=1}^{L} \{\textproc{K-Means++}(C_{\ell}, N \left(\frac{|C_\ell|}{|C|}\right), d)\}$ 
    \State{\textbf{Return} $\Xi$}
    \EndProcedure
    \end{algorithmic}
\end{algorithm}

\subsection{CPO: Projection}
After obtaining $\Xi$, further insight can be gleaned by exploring the local projection around each $\xi^{(i)}$. An example of this is visualizing the road-level variability in traffic predictions from uncertainty in upstream weather predictions, shown in Figure \ref{fig:travel_times}.
To do this, we visualize the extent of the Voronoi cell $V^{(i)}\subset\mathcal{C}(x)$ associated with $\xi^{(i)}$ along the $\mathcal{C}$ space dimensions. 
That is, for each Voronoi cell, we visualize the Frechet variance along the projections $\{\pi_j\}_{j=1}^J$, where $J = \text{dim}(\mathcal{C})$. Such projections preserve the structure of the objects being modeled, making them visually interpretable. For instance, $\pi_j$ in the traffic example corresponds to the projection of $V^{(i)}$ to a \textit{single} road $j$. Similarly, $\pi_j$ would project to a single atom for a molecular reconstruction task. 
Formally,
\begin{equation}
    \left|V^{(i)}_j\right| := \sum_{c\in V^{(i)}}  d^2(\pi_{j}(c), \pi_{j} (\xi^{(i)})).
\end{equation}


\section{EXPERIMENT}\label{section:experiments}
We now demonstrate the utility of the CPO framework. Code will be made public upon acceptance. 


\subsection{SBI: Fractional Knapsack}\label{section:experiments_sbi}
We first study the fractional knapsack problem under various complex contextual mappings, namely
\begin{gather}
w^{*}(x) := \min_{w,\mathcal{U}} \max_{\widehat{c}\in\mathcal{U}(x)} \quad -\widehat{c}^{T} w \\
\textrm{s.t.} w\in [0,1]^{n}, p^{T} w \le B, \mathcal{P}_{X,C}(C\in\mathcal{U}(X)) \ge 1-\alpha, \nonumber
\end{gather}
where $p\in\mathbb{R}^{n}$ and $B > 0$. The distributions $\mathcal{P}(C)$ and $\mathcal{P}(X\mid C)$ are taken to be those from various simulation-based inference (SBI) benchmark tasks provided by \citep{hermans2021averting}, chosen as they have $\mathcal{P}(C\mid X)$ with complex structure. We specifically study Two Moons, Lotka-Volterra, Gaussian Linear Uniform, Bernoulli GLM, Susceptible-Infected-Recovered (SIR), and Gaussian Mixture, fully described in Appendix \ref{section:benchmark}. We note that, while these particular distributions have little semantic meaning in the traditional context of fractional knapsack, this experiment highlights the capacity for CPO to succeed even for complex distributions, which we leverage in a more semantically meaningful case in Section \ref{section:experiments_routing}.

\subsubsection{SBI: Quantitative Assessment}


\begin{table*}[t]
\caption{\label{table:sbi_results} Coverages across tasks for $\alpha=0.05$ are shown in the left table, where coverage was assessed over a batch of 1,000 i.i.d. test samples. Objective optima are shown in the right table, averaged over a batch of 10 i.i.d. test samples with standard deviations in parentheses. The nominal optima are included as reference points.}
\resizebox{\textwidth}{!}{%
\begin{tabular}[t]{cccccc}
\toprule
& Box & PTC-B & Ellipsoid & PTC-E & CPO \\
\midrule
Gaussian Uniform & 0.94 & 0.96 & 0.95 & 0.95 & 0.95 \\
Gaussian Mixture & 0.95 & 0.93 & 0.94 & 0.93 & 0.94 \\
Bernoulli GLM & 0.96 & 0.95 & 0.95 & 0.94 & 0.94 \\
Lotka Volterra & 0.95 & 0.96 & 0.94 & 0.94 & 0.95 \\
SIR & 0.94 & 0.95 & 0.93 & 0.95 & 0.93 \\
Two Moons & 0.93 & 0.94 & 0.94 & 0.94 & 0.96 \\
\bottomrule
\end{tabular}
    \quad
\begin{tabular}[t]{cccccc}
\toprule
         Box &        PTC-B &    Ellipsoid &        PTC-E &                   CPO &                Nominal \\
\midrule
   0.0 (0.0) &    0.0 (0.0) &    0.0 (0.0) & \textbf{-0.27 (0.35)} &  \textbf{-0.43 (0.4)} &  \textit{-4.48 (0.56)} \\
   0.0 (0.0) &  \textbf{-6.6 (1.67)} &    0.0 (0.0) & \textbf{-7.38 (1.78)} & \textbf{-7.77 (1.87)} & \textit{-11.66 (1.23)} \\
   0.0 (0.0) & \textbf{-0.18 (0.49)} &    0.0 (0.0) & \textbf{-0.06 (0.25)} & \textbf{-0.18 (0.37)} &  \textit{-3.53 (0.27)} \\
\textbf{-0.52 (0.02)} & -0.05 (0.24) &  -0.02 (0.0) & -0.22 (0.18) & \textbf{-0.68 (0.26)} &  \textit{-1.88 (0.01)} \\
-0.16 (0.02) & -0.22 (0.09) & -0.08 (0.01) & -0.22 (0.06) & \textbf{-0.38 (0.05)} &  \textit{-0.52 (0.02)} \\
   0.0 (0.0) &    0.0 (0.0) &    0.0 (0.0) &    0.0 (0.0) & \textbf{-0.15 (0.11)} &  \textit{-0.38 (0.01)} \\
\bottomrule
\end{tabular}
}
\end{table*}

We first demonstrate the quantitative improvement in decision-making from leveraging CPO over the box- (PTC-B) and ellipsoid-based (PTC-E) regions proposed in \citep{sun2023predict}, as well as box- and ellipsoid-based sets constructed based solely on observations of $\mathcal{P}(C)$, i.e. where we ignore $x$, referred to as Box and Ellipsoid. For CPO, we use
\begin{equation}\label{eqn:sbi_score}
    s(x,c) = \min_{k} || \widehat{c}_k - c ||_2^2.
\end{equation}
$q(\widehat{c} \mid x)$ was taken to be a neural spline normalizing flow \citep{durkan2019neural} trained with FAVI \citep{ambrogioni2019forward}. Visualizations of the exact and variational posteriors are provided in Appendix \ref{section:posteriors}. $K$s were chosen by studying the inflections of the prediction region volume estimate under each distributional setup, with $|\mathcal{D}_{\mathcal{C}_{1,2}}| = 1000$, seen in Figure \ref{fig:sbi_results_volume}. Inflection points were around $K=10$ for most setups.

For assessing coverage and the robust objective value, we sampled $|\mathcal{D}_{\mathcal{T}}| = 1000$ test points i.i.d. from $\mathcal{P}(X,C)$. Coverage was assessed across all 1000 samples by measuring the proportion of samples for which $s(x^{(i)},c^{(i)})\le\widehat{q}$. For assessing the objective, optimization was performed across 10 samples, with $p\sim U([0, 1000]^{n})$, $u\sim U(0, 1)$, and $B\sim U(\max_{i} p_{i}, \sum_i p_{i} - u \max_{i} p_{i})$ sampled per run.

\begin{figure}[H]
\centering
\includegraphics[scale=0.20]{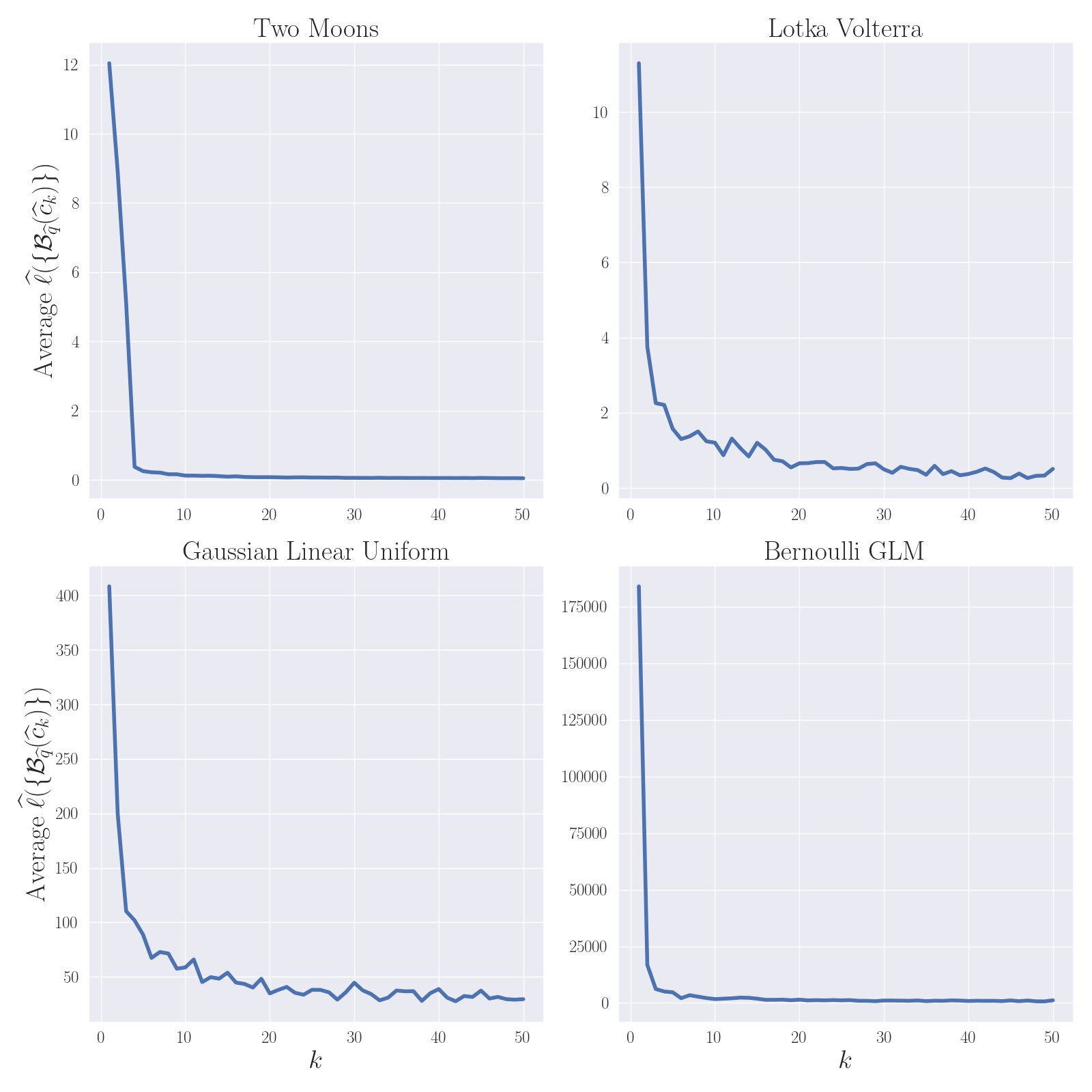}
\caption{\label{fig:sbi_results_volume} Average volume estimates $\widehat{\ell}(\{\mathcal{B}_{\widehat{q}}(\widehat{c}_{k}^{(i)})\})$ over $x^{(i)}\in\mathcal{D}_{\mathcal{C}_2}$ across SBI benchmarks.}
\end{figure}


The results are seen in Table \ref{table:sbi_results}. We include the nominal optima as a reference, i.e. $\min_{w} -c^{T} w$ for the \textit{true} $c$. Recall that, by Lemma \ref{lemma:coverage_bound}, with proper $\mathcal{U}(x)$, the robust objective values should be valid upper bounds on the nominal optima, with more conservative regions resulting in more vacuous bounds. We see this as, although all approaches result in valid coverage guarantees and hence produce valid upper bounds, the overly conservative nature of alternate regions results in their consistent looseness compared to CPO. Notably, these differences are more accentuated in cases where $\mathcal{P}(C|X)$ has complex structure; level sets under the Gaussian Linear, Gaussian Mixture, and Bernoulli GLM cases are roughly ellipsoidal, seen in Appendix \ref{section:posteriors}, resulting in comparable performance between CPO and PTC-E. Thus, as discussed and highlighted in Section \ref{section:experiments_routing}, the benefits of CPO primarily manifest under difficult-to-model contextual distributions, where sets for simple geometries become overly large.


\subsubsection{SBI: Representative Point Recovery}
We next demonstrate that Algorithm \ref{alg:CPO-RPs} can approximately recover RPs for such uncertainty regions, leveraged to glean insights in the modeling task of Section \ref{section:experiments_routing}. 
Notably, RPs are not unique; for instance, any rigid rotation of $\Xi$ for a uniform distribution over a 2D ball results in a distinct yet optimal set $\widehat{\Xi}$ of RPs. The RP objective minimum, however, is unique, meaning suboptimality can be assessed by measuring
\begin{equation}
    \Delta(\Xi, \widehat{\Xi}) := \mathbb{E}_{C\sim U(\mathcal{C}(x))}\left[d(C, \widehat{\Xi}) - d(C, \Xi) \right].
\end{equation}
$N=5$ representative points were produced per setup. To compute $\Xi$, a grid discretization over the space was performed followed by a clustering for each connected component of this discretization. That is, the support $\mathcal{C}$ was discretized into $60$ bins per dimension. Each discretized point $c_k$ was assessed for membership in $\mathcal{C}(x)$, resulting in a collection of points $C$, from which we could
recover $\Xi$ in the manner described in Section \ref{section:method_rps}. Visualizations of the exact and approximate RPs are provided for tasks where $\mathcal{C}\subset\mathbb{R}^{2}$ in Appendix \ref{appendix:sbi_rps}.

\begin{figure}[H]
\centering
\includegraphics[scale=0.20]{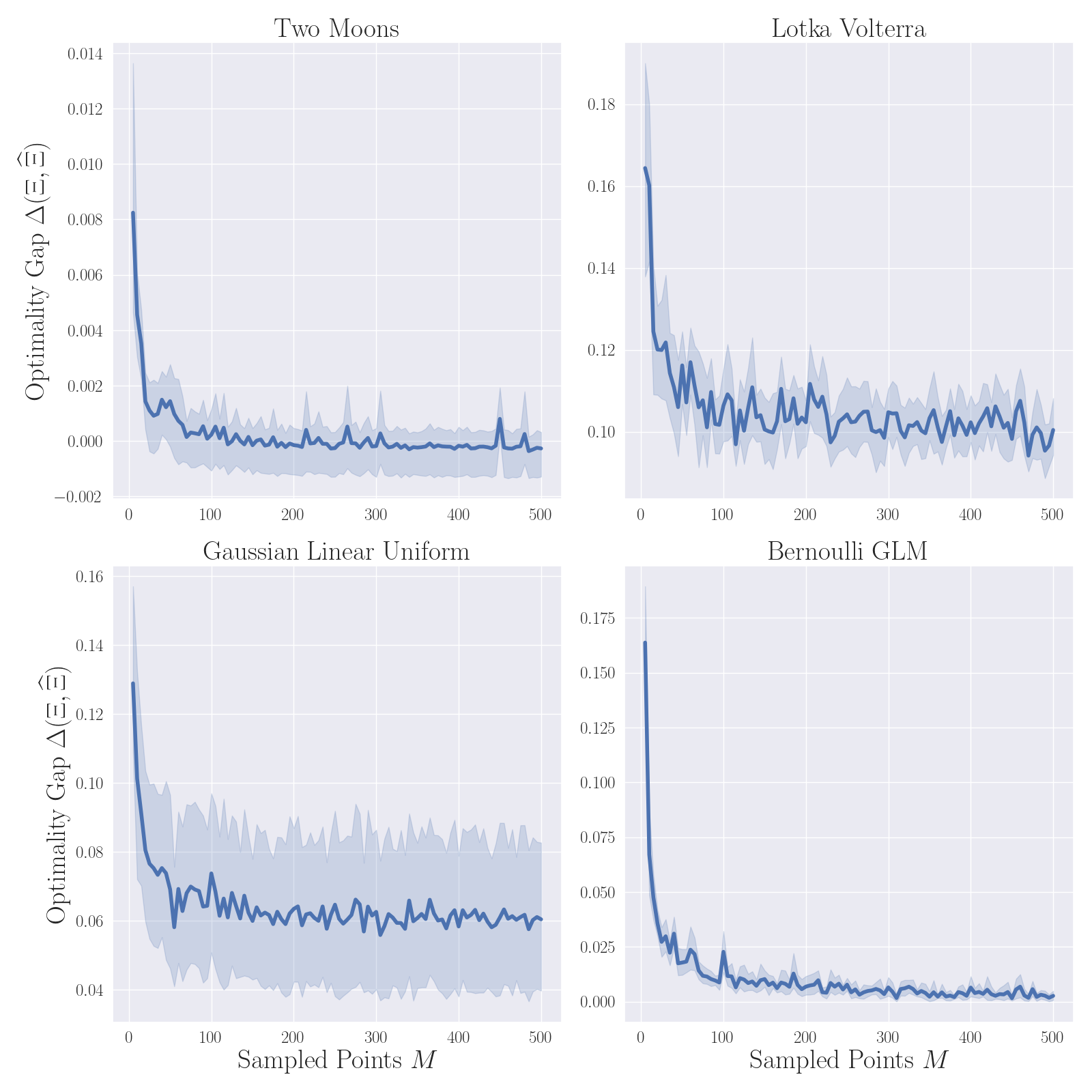}
\caption{\label{fig:sbi_results_rp} Suboptimality of the approximate representative points $\Delta(\Xi, \widehat{\Xi})$ decreases over increased sampling from the conformal prediction region. }
\end{figure}

To make explicit discretization possible, problems were projected into lower-dimensional versions, namely $\mathcal{C}\subset\mathbb{R}^4$. Figure \ref{fig:sbi_results_rp} demonstrates the suboptimality of $\widehat{\Xi}$ decreases with increasing samples. Of note is that this convergence is slower in higher dimensional problems: for low dimensional cases, recovery of optimal RPs happens for small $M$, meaning any fluctuations thereafter are noise, as seen in the Two Moons case.

\begin{table*}[t]
\caption{\label{table:vehicle_results} Coverage was assessed over 128 i.i.d. test samples and average objective optima over 10 i.i.d. test samples with standard deviations in parentheses.}
\resizebox{\textwidth}{!}{%
\begin{tabular}[t]{ccccccc}
\toprule
& Box & PTC-B & Ellipsoid & PTC-E & CPO & \textit{Nominal} \\
\midrule
Coverage & 0.94 & 0.93 & 0.94 & 0.92 & 0.94 & --- \\
Objective & 7863.45 (0.0) & 34559.03 (171.3) & 7038.77 (0.0) & 8807.68 (4.22) & \textbf{4171.22 (321.34)} & \textit{299.50 (0.0)} \\
\bottomrule
\end{tabular}
}
\end{table*}

\subsection{Robust Vehicle Routing}\label{section:experiments_routing}
Optimal routing is a long-standing point of interest in the operations research community, with widespread applications such as in resource distribution and urban traffic flow management \citep{mor2022vehicle, saberi2012continuous, okulewicz2019metaheuristic, kovrenavr2003vehicle}. 
We study the traffic flow problem from \citep{angelelli2021system}.

Recent work has demonstrated the utility of generative models in quantifying uncertainty for weather predictions over traditional physics-based approaches \citep{agrawal2019machine,ayzel2020rainnet,franch2020precipitation,shi2017deep}. We specifically leverage a latent diffusion model for such forecasting from \citep{leinonen2023latent}. Formally, a forecaster $\mathcal{P}(\widetilde{Y} \mid x)$ maps precipitation readings from radar networks $x\in\mathbb{R}^{T\times W\times H}$, specifically over $T$ time steps with resolutions $W\times H$, to $\widetilde{Y}\in \mathbb{R}^{W\times H}$, the precipitation for some fixed $\Delta T$ point beyond $x$.

We consider the robust traffic flow problem (RTFP) for a source-target pair $(s,t)$ over the network graph of Manhattan, where $|\mathcal{V}| = 4584$ and $|\mathcal{E}| = 9867$. The precipitation $\widetilde{Y}$ was combined with the nominal speed limits to obtain the final travel costs $c$ along edges, fully described in Appendix \ref{appendix:experiments_routing}. Formally, we seek
\begin{gather}\label{equation:flow_lp}
w^{*}(x) := \min_{w} \max_{\widehat{c}\in\mathcal{U}(x)} \widehat{c}^{T} w \\
\textrm{s.t.} w\in [0,1]^{\mathcal{E}}, Aw = b, \mathcal{P}_{X,C}(C\in\mathcal{U}(X)) \ge 1-\alpha \nonumber
\end{gather}
where $w_{e}$ represents the proportion of traffic routed along edge $e$, $C\in\mathbb{R}^{|\mathcal{E}|}$ is the edge weight vector, $A\in\mathbb{R}^{|\mathcal{V}|\times |\mathcal{E}|}$ is the node-arc incidence matrix, and $b\in\mathbb{R}^{|\mathcal{V}|}$ has entries $b_{s} = 1, b_{t} = -1,$ and $b_{k} = 0$ for $k\notin\{s, t\}$.


We again demonstrate the quantitative improvement in decision-making resulting from using the more informative CPO prediction regions. Experiments were conducted with $s$ and $t$ chosen uniformly at random from $\mathcal{V}$. We take the score as defined in Equation \ref{eqn:sbi_score} on the edge weight space rather than the initial precipitation map space. Results are shown in Table \ref{table:vehicle_results}. Again, although all approaches achieve coverage guarantees, bounds resulting from alternate regions are significantly looser compared to those from CPO. This is especially prominent in this task compared to those of Section \ref{section:experiments_sbi} due to the high dimension of the prediction space($\mathbb{R}^{|\mathcal{E}|}$) and complex nature of $\mathcal{P}(C|X)$.

\begin{figure}[H]
\centering
\includegraphics[scale=0.35]{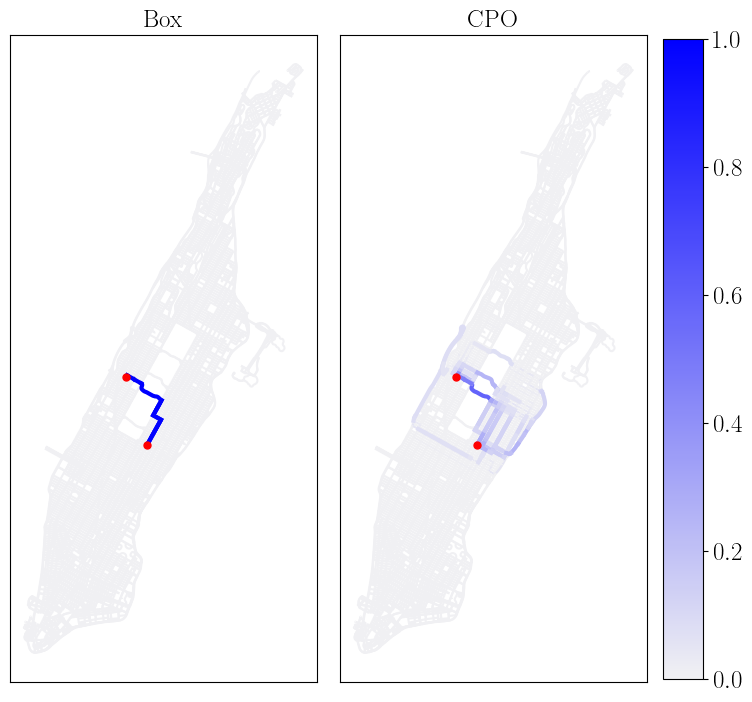}
\caption{\label{fig:robust_traffic_routing} Solutions for the RTFP under the Box (left) and CPO (right) uncertainty regions.}
\end{figure}

Notably, the formulation in Equation \ref{equation:flow_lp} is a relaxation of the standard LP formulation of the robust shortest paths problem (RSPP), in which $\mathcal{W} = \{0,1\}^{\mathcal{E}}$. Given that $A$ is a totally unimodular matrix, the solutions of the \textit{box}-constrained RTFP and RSPP are equivalent, i.e. for both Box and PTC-B; they, however, are \textit{not} equivalent under more general constraint sets \citep{chaerani2005robust}, i.e. Ellipsoid, PTC-E, and CPO, resulting in the observed suboptimality of box constraints. This is highlighted in Figure \ref{fig:robust_traffic_routing}, where the Box constraint results in a fully concentrated allocation of traffic along a single path.

Despite apparent quantitative improvements resulting from the CPO optimal solution, it is difficult to directly understand \textit{why} such allocations were deemed optimal without a qualitative impression of $\mathcal{U}(x)$, as framed in Section \ref{section:method_rps}. We, therefore, now construct $N=5$ representative points and their corresponding projections, two of which are visualized in Figure \ref{fig:travel_times}. 
The RPs highlight the multimodal nature of the edge weights distribution, where $\xi^{(1)}$ exhibits a case of precipitation more heavily concentrating along the northeast corridor across Manhattan and $\xi^{(2)}$ one where it concentrates on the west. In addition, the projection around $\xi^{(2)}$ reveals especially high uncertainty on the path through Central Park with less on surrounding roads. CPO, thus, hedges its allocation in Figure \ref{fig:robust_traffic_routing} more evenly across paths, unlike the concentrated allocation under the Box region.

\begin{figure}[H]
\centering
\includegraphics[scale=0.20]{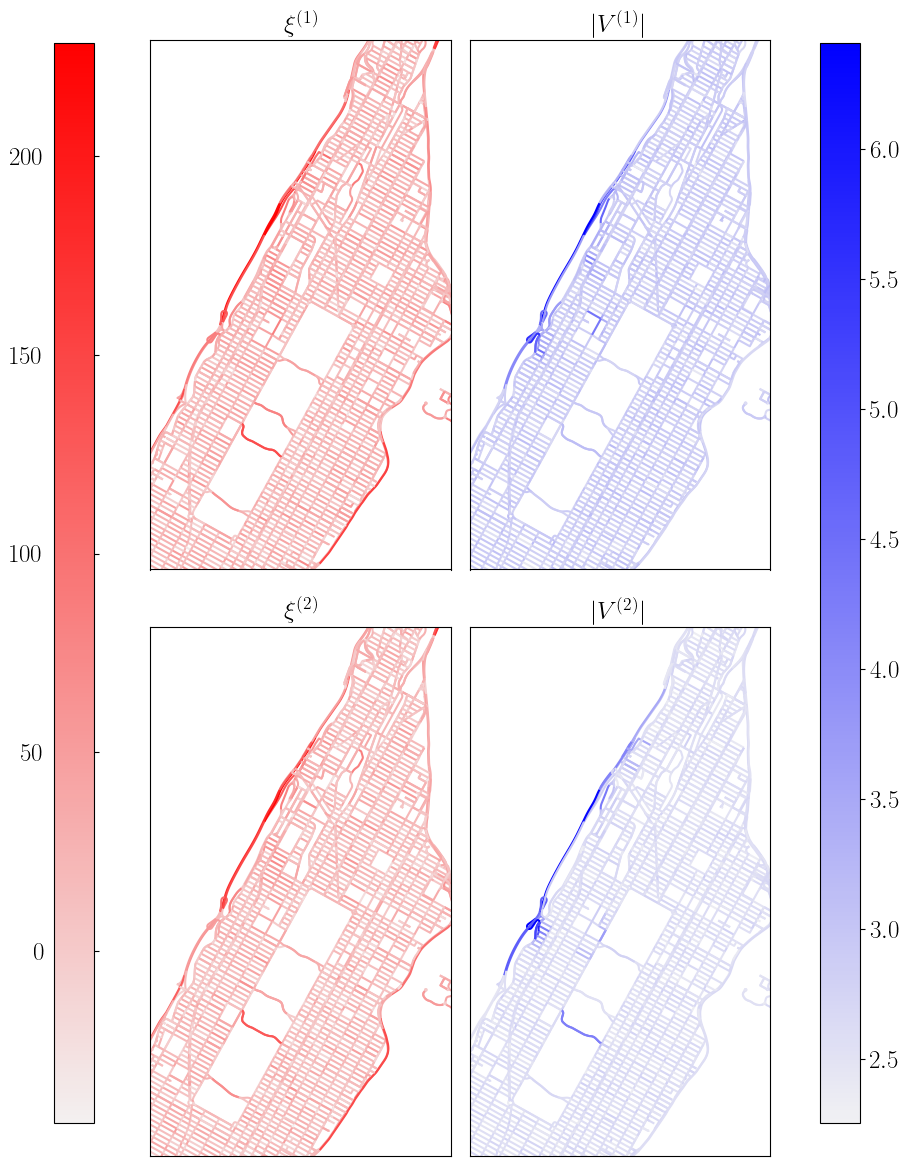}
\caption{\label{fig:travel_times} Two RPs for $\mathcal{C}(x)$ for travel time prediction (left) and the extents of their Voronoi cells (right).}
\end{figure}

\section{DISCUSSION}\label{section:discussion} 
We have presented \textproc{CPO}, a framework to leverage informative, non-convex conformal regions for predict-then-optimize decision-making. 
This work suggests many directions for future work. We are pursuing the use of CPO for \textit{sequential} decision-making, where non-exchangeable conformal prediction is required to handle sampling \citep{fannjiang2022conformal}. Another interesting extension would be applications of CPO to discrete objects using GFlowNets for conditional sampling \citep{malkin2022gflownets,hu2023gflownet}. 
Finally, leveraging CPO over function spaces would enable its use to distributionally robust optimization.

\clearpage
\printbibliography

\newpage
\section*{Checklist}

 \begin{enumerate}

 \item For all models and algorithms presented, check if you include:
 \begin{enumerate}
   \item A clear description of the mathematical setting, assumptions, algorithm, and/or model. [Yes]
   \item An analysis of the properties and complexity (time, space, sample size) of any algorithm. [Yes]
   \item (Optional) Anonymized source code, with specification of all dependencies, including external libraries. [No, will be released upon acceptance]
 \end{enumerate}

 \item For any theoretical claim, check if you include:
 \begin{enumerate}
   \item Statements of the full set of assumptions of all theoretical results. [Yes]
   \item Complete proofs of all theoretical results. [Yes]
   \item Clear explanations of any assumptions. [Yes]     
 \end{enumerate}

 \item For all figures and tables that present empirical results, check if you include:
 \begin{enumerate}
   \item The code, data, and instructions needed to reproduce the main experimental results (either in the supplemental material or as a URL). [Yes]
   \item All the training details (e.g., data splits, hyperparameters, how they were chosen). [Yes]
         \item A clear definition of the specific measure or statistics and error bars (e.g., with respect to the random seed after running experiments multiple times). [Yes]
         \item A description of the computing infrastructure used. (e.g., type of GPUs, internal cluster, or cloud provider). [Yes]
 \end{enumerate}

 \item If you are using existing assets (e.g., code, data, models) or curating/releasing new assets, check if you include:
 \begin{enumerate}
   \item Citations of the creator If your work uses existing assets. [Not Applicable]
   \item The license information of the assets, if applicable. [Not Applicable]
   \item New assets either in the supplemental material or as a URL, if applicable. [Not Applicable]
   \item Information about consent from data providers/curators. [Not Applicable]
   \item Discussion of sensible content if applicable, e.g., personally identifiable information or offensive content. [Not Applicable]
 \end{enumerate}

 \item If you used crowdsourcing or conducted research with human subjects, check if you include:
 \begin{enumerate}
   \item The full text of instructions given to participants and screenshots. [Not Applicable]
   \item Descriptions of potential participant risks, with links to Institutional Review Board (IRB) approvals if applicable. [Not Applicable]
   \item The estimated hourly wage paid to participants and the total amount spent on participant compensation. [Not Applicable]
 \end{enumerate}

 \end{enumerate}

\newpage
\onecolumn
\appendix




\section{Prediction Region Validity Lemma}\label{section:coverage_bound}
\begin{lemma}
    Consider any $f(w, c)$ that is $L$-Lipschitz in $c$ under the metric $d$ for any fixed $w$. Assume further that $\mathcal{P}_{X,C}(C \in \mathcal{U}(X)) \ge 1 - \alpha$. Then, 

    \begin{equation}
        \mathcal{P}_{X,C}\left(\Delta(X, C) \le L \mathrm{\ diam}(\mathcal{U}(X))\right) \ge 1 - \alpha.
    \end{equation}
\end{lemma}

\begin{proof}
    We consider the event of interest conditionally on a pair $(x, c)$ where $c\in\mathcal{U}(x)$:
    
    \begin{gather*}
        \min_{w} \max_{\widehat{c}\in\mathcal{U}(x)} f(w, \widehat{c}) - \min_{w} f(w, c) \\
        \le \max_{w} | \max_{\widehat{c} \in \mathcal{U}(x)} f(w, \widehat{c}) - f(w, c) | \\ 
        \le L \max_{\widehat{c} \in \mathcal{U}(x)} d(\widehat{c}, c) 
        \le L \mathrm{diam}(\mathcal{U}(x)).
    \end{gather*}

    Since we have the assumption that $\mathcal{P}(C \in \mathcal{U}(X)) \ge 1 - \alpha$, the result immediately follows.
\end{proof}

\section{Optimization Convergence Lemma}\label{section:convergence}
We first begin by citing a standard result of projected gradient descent, from which the result of interest immediately follows.

\begin{lemma}\label{lemma:pgd_convergence}
    Let $K$ be a closed convex set, and $f : K \rightarrow \mathbb{R}$ be convex, differentiable, and $L$-Lipschitz. Let $x^{*}\in K$ be a minimizer of $f$, and define $T := \frac{L^2 || x_0 - x^{*} ||}{\epsilon^2}$ and $\eta := \frac{|| x_0 - x^{*} ||}{L \sqrt{T}}$. Then the iterates $\{x_{t}\}_{t=0}^{T}$ returned by projected gradient descent satisfy

    \begin{equation}
        f\left(\frac{1}{T+1}\sum_{t=0}^{T} x_{t}\right) - f(x^{*}) \le \epsilon.
    \end{equation}
\end{lemma}

\begin{lemma}
    Let $\phi(w) := \max_{\widehat{c} \in \bigcup_{k=1}^{K} \mathcal{B}_{\widehat{q}}(\widehat{c}_{k})} f(w, \widehat{c})$ for $\{\widehat{c}_{k}\}_{k=1}^{K}\subset\mathcal{C}$, $\widehat{q}\in\mathbb{R}^{+}$, and $f(w, c)$ convex-concave and $L$-Lipschitz in $c$ for any fixed $w$. Let $w^{*}\in \mathcal{W}$ be a minimizer of $\phi$. For any $\epsilon > 0$, define $T := \frac{L^2 || w_0 - w^{*} ||}{\epsilon^2}$ and $\eta := \frac{|| w_0 - w^{*} ||}{L \sqrt{T}}$. Then the iterates $\{w_{t}\}_{t=0}^{T}$ returned by Algorithm \ref{alg:CPO-Opt} satisfy

    \begin{equation}
        \phi\left(\frac{1}{T+1}\sum_{t=0}^{T} w_{t}\right) -  \phi(w^{*}) \le \epsilon.
    \end{equation}
\end{lemma}

\begin{proof}
    Notice that $\phi(w)$ is convex by Danskin's Theorem by assumption of the convexity of $f$ in $w$. By Danskin's Theorem, $\nabla_{w} \phi(w) = \nabla_{w} f(w, c^{*})$, where $c^{*} := \max_{\widehat{c} \in \mathcal{C}(x)} f(w, \widehat{c})$. Further notice
    
    \begin{equation}
        \phi(w) 
        := \max_{\widehat{c} \in \mathcal{C}(x)} f(w, \widehat{c})
        = \max_{k} \max_{\widehat{c} \in \mathcal{B}_{\widehat{q}}(\widehat{c}_{k})} f(w, \widehat{c}).
    \end{equation}

    Denote $\phi_{k}(w) := \max_{\widehat{c} \in \mathcal{B}_{\widehat{q}}(\widehat{c}_{k})} f(w, \widehat{c})$. Clearly, $\phi_{k}(w)$ is $L$-Lipschitz by assumption on the structure of $f$. Further, as the point-wise maximum of $L$-Lipschitz functions is itself $L$-Lipschitz, it follows that $\phi(w) = \max_{k} \phi_{k}(w)$ is also $L$-Lipschitz. The conclusion, thus, follows by applying Lemma \ref{lemma:pgd_convergence} to $\phi(w)$.
\end{proof}

\section{Simulation-Based Inference Benchmarks}\label{section:benchmark}
The benchmark tasks are a subset of those provided by \citep{lueckmann2021benchmarking}. For convenience, we provide brief descriptions of the tasks curated by this library; however, a more comprehensive description of these tasks can be found in their manuscript.

\subsection{Gaussian Linear}
10-dimensional Gaussian model with a Gaussian prior:

\begin{gather*}
    \text{\textbf{Prior}: } \mathcal{N}(0, 0.1 \odot I)\\
    \text{\textbf{Simulator}: } x\mid w \sim \mathcal{N}(x\mid w, 0.1 \odot I)
\end{gather*}

\subsection{Gaussian Linear Uniform}
10-dimensional Gaussian model with a uniform prior:

\begin{gather*}
    \text{\textbf{Prior}: } \mathcal{U}(-1, 1)\\
    \text{\textbf{Simulator}: } x\mid w \sim \mathcal{N}(x\mid w, 0.1 \odot I)
\end{gather*}



\subsection{SLCP with Distractors}
Simple Likelihood Complex Posterior (SLCP) with Distractors has uninformative dimensions in the observation over the standard SLCP task:

\begin{gather*}
    \text{\textbf{Prior}: } \mathcal{U}(-3, 3)  \\
    \text{\textbf{Simulator}: } x\mid w = p(y)
    \text{ where } p \text{ reorders } \\
    y \text{ with a fixed random order } \\
    y_{[1:8]} \sim \mathcal{N}\left(
        \begin{bmatrix}
        w_1 \\ 
        w_2
        \end{bmatrix}, 
    \begin{bmatrix}
        w_3^4 & w_3^2 w_4^2 \tanh(w_5) \\ 
        w_3^2 w_4^2 \tanh(w_5) & w_4^4
    \end{bmatrix}\right),\\
    y_{9:100} \sim \frac{1}{20} \sum_{i=1}^{20} t_2(\mu^i, \Sigma^i), 
    \mu^i\sim\mathcal{N}(0, 15^2 I),  \\
    \Sigma^i_{j,k}\sim\mathcal{N}(0, 9),
    \Sigma^i_{j,j} = 3 e^{a}, a\sim\mathcal{N}(0, 1),
\end{gather*}



\subsection{Bernoulli GLM Raw}
10-parameter GLM with Bernoulli observations and Gaussian prior. Observations are not sufficient statistics, unlike the standard ``Bernoulli GLM'' task:

\begin{gather*}
    \text{\textbf{Prior}: } 
    \beta\sim\mathcal{N}(0, 2), f\sim\mathcal{N}(0, (F^T F)^{-1}) \\
    \qquad F_{i,i-2} = 1, F_{i,i-1} = -2 \\
    F_{i,i} = 1 + \sqrt{\frac{i-1}{9}}, F_{i, j} = 0; i\le j \\
    \text{\textbf{Simulator}: } x^{(i)} \mid w \sim \text{Bern}(\eta(v^{(i)}_T f+ \beta)), \\ \eta(\odot) = \exp(\odot) / (1 + \exp(\odot)) 
\end{gather*}

\subsection{Gaussian Mixture}
A mixture of two Gaussians, with one having a much broader covariance structure:

\begin{gather*}
    \text{\textbf{Prior}: } 
    \beta\sim\mathcal{U}(-10, 10) \\
    \text{\textbf{Simulator}: } x \mid w \sim 0.5 \mathcal{N}(x\mid w, I) + 0.5 \mathcal{N}(x\mid w, .01 I)
\end{gather*}

\subsection{Two Moons}
Task with a posterior that has both global (bimodal) and local (crescent-shaped) structure:

\begin{gather*}
    \text{\textbf{Prior}: } 
    \beta\sim\mathcal{U}(-1, 1) \\
    \text{\textbf{Simulator}: } x \mid w =  \\
    \begin{bmatrix}
        r \cos(\alpha) + 0.25 \\ 
        r \sin(\alpha)
    \end{bmatrix} +
    \begin{bmatrix}
        -|w_1 + w_2| / \sqrt{2} \\ 
        (-w_1 + w_2) / \sqrt{2}
    \end{bmatrix}\\
    \qquad \alpha\sim \mathcal{U}(-\pi/2,\pi/2), r\sim\mathcal{N}(0.1,0.01^2)
\end{gather*}

\subsection{SIR}
Epidemiology model with $S$ (susceptible), $I$ (infected), and $R$ (recovered). A contact rate $\beta$ and mean recovery rate of $\gamma$ are used as follows:

\begin{gather*}
    \text{\textbf{Prior}: } 
    \beta\sim\text{LogNormal}(\log(0.4),0.5), \\
    \gamma\sim\text{LogNormal}(\log(1/8),0.2) \\
    \text{\textbf{Simulator}: } x = (x^{(i)})_{i=1}^{10}; x^{(i)} \mid w\sim\text{Bin}(1000, \frac{I}{N}),  \\
    \text{ where } I \text{ is simulated from: } 
    \\
    \frac{dS}{dt} = -\beta \frac{SI}{N}, \qquad
    \frac{dI}{dt} = \beta \frac{SI}{N} - \gamma I, \qquad
    \frac{dR}{dt} = \gamma I
\end{gather*}

\subsection{Lotka-Volterra}
An ecological model commonly used in describing dynamics of competing species. $w$ parameterizes this interaction as $w = (\alpha, \beta, \gamma, \delta)$:

\begin{gather*}
    \text{\textbf{Prior}: } 
    \alpha\sim\text{LogNormal}(-.125,0.5) \\
    \beta\sim\text{LogNormal}(-3,0.5),
    \gamma\sim\text{LogNormal}(-.125,0.5) \\
    \delta\sim\text{LogNormal}(-3,0.5) \\
    \text{\textbf{Simulator}: } x = (x^{(i)})_{i=1}^{10},  \\
    x_{1,i} \mid w\sim\text{LogNormal}(\log(X), 0.1),  \\
    x_{2,i} \mid w\sim\text{LogNormal}(\log(Y), 0.1) \\ 
    \text{ where } X,Y \text{ is simulated from: }  \\
    \frac{dX}{dt} = \alpha X - \beta X Y, \qquad
    \frac{dY}{dt} = - \gamma Y + \delta XY
\end{gather*}

\section{Training Details}\label{section:exp_details}
All encoders were implemented in PyTorch \citep{paszke2019pytorch} with a Neural Spline Flow architecture. The NSF was built using code from \citep{nflows}. Specific architecture hyperparameter choices were taken to be the defaults from \citep{nflows} and are available in the code. Optimization was done using Adam \citep{kingma2014adam} with a learning rate of $10^{-3}$ over 5,000 training steps. Minibatches were drawn from the corresponding prior $\mathcal{P}(Y)$ and simulator $\mathcal{P}(X\mid Y)$ as specified per task in the preceding section. Training these models required between 10 minutes and two hours using an Nvidia RTX 2080 Ti GPUs for each of the SBI tasks.

\section{Posteriors}\label{section:posteriors}
We provide visualizations of approximate and reference posteriors (produced with MCMC from \citep{lueckmann2021benchmarking}).

\subsection{Gaussian Linear}
\begin{figure}[H]
\includegraphics[scale=0.1]{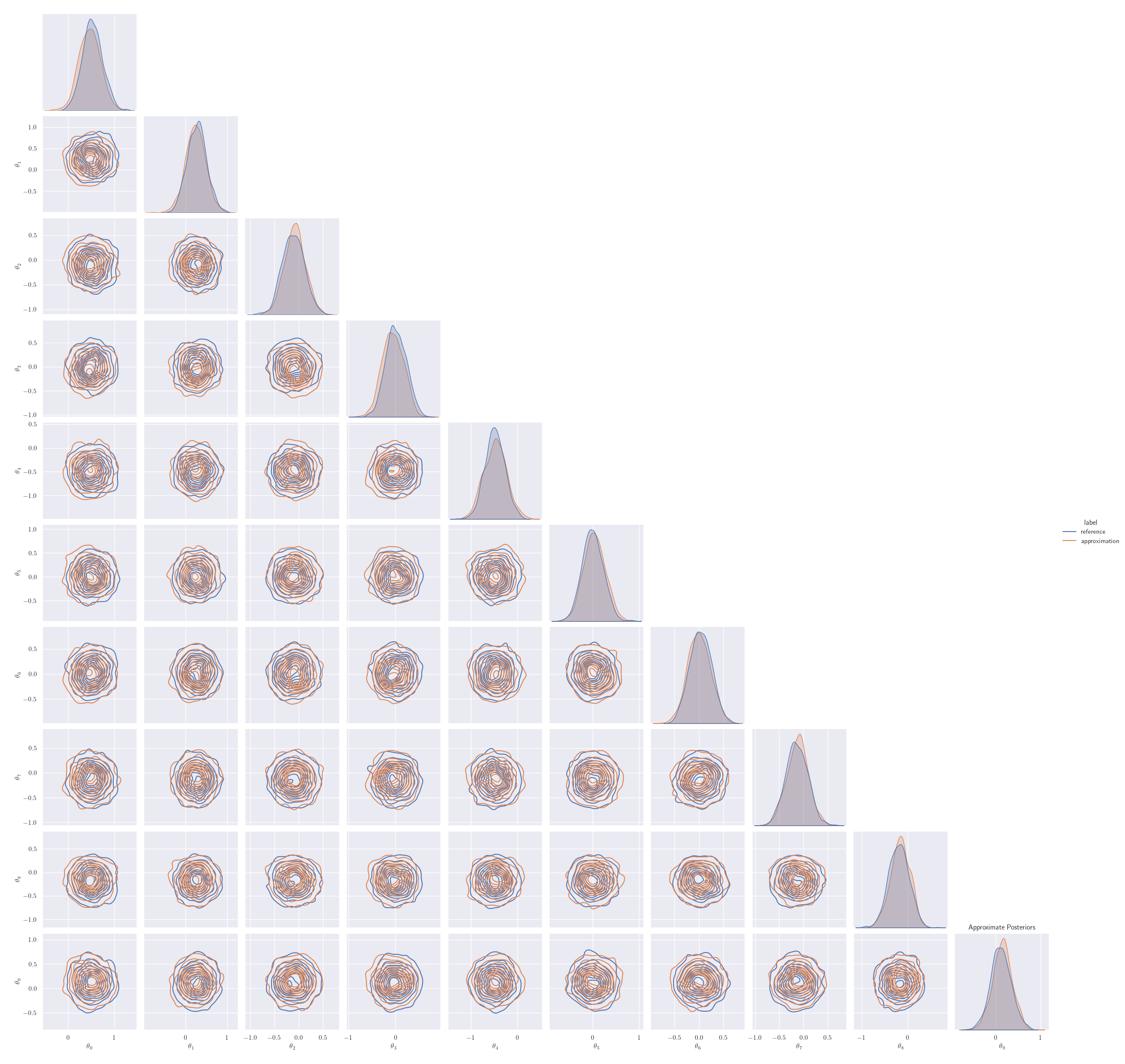}
\centering
\end{figure}

\subsection{Gaussian Mixture}
\begin{figure}[H]
\includegraphics[scale=0.4]{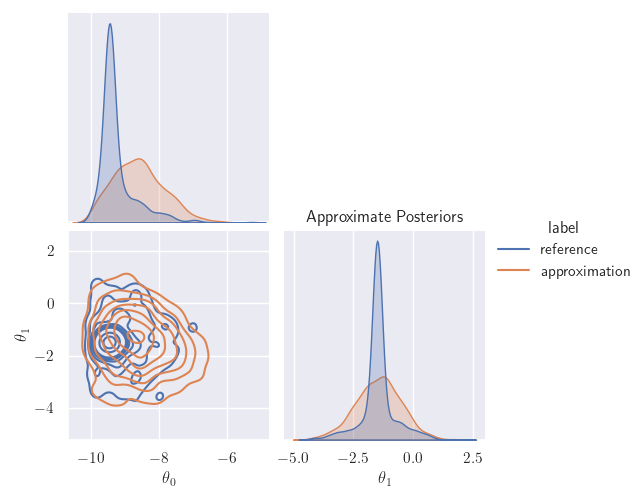}
\centering
\end{figure}

\subsection{Gaussian Linear Uniform}
\begin{figure}[H]
\includegraphics[scale=0.1]{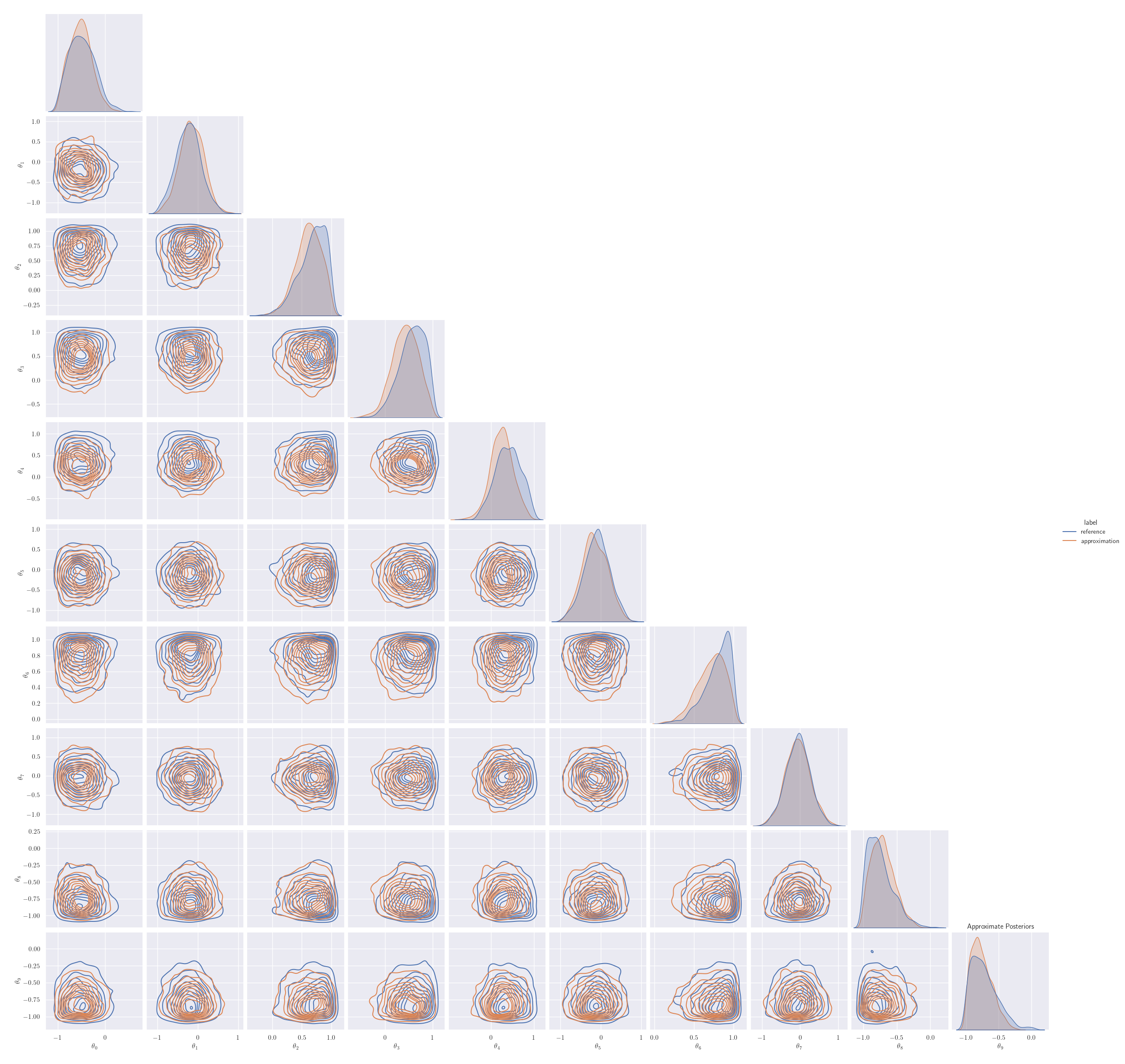}
\centering
\end{figure}

\subsection{Two Moons}
\begin{figure}[H]
\includegraphics[scale=0.4]{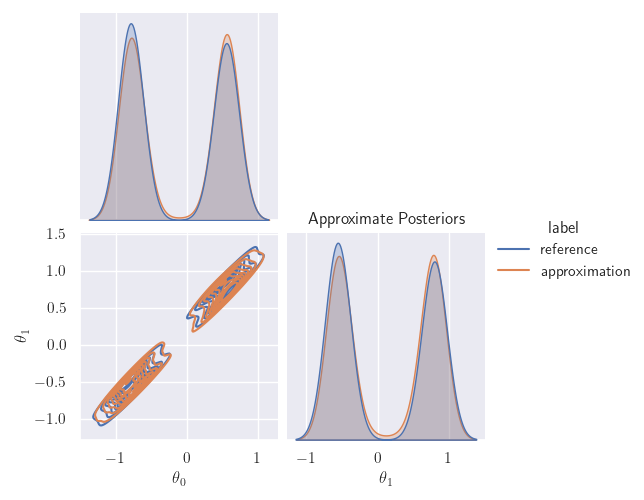}
\centering
\end{figure}

\subsection{SLCP}
\begin{figure}[H]
\includegraphics[scale=0.20]{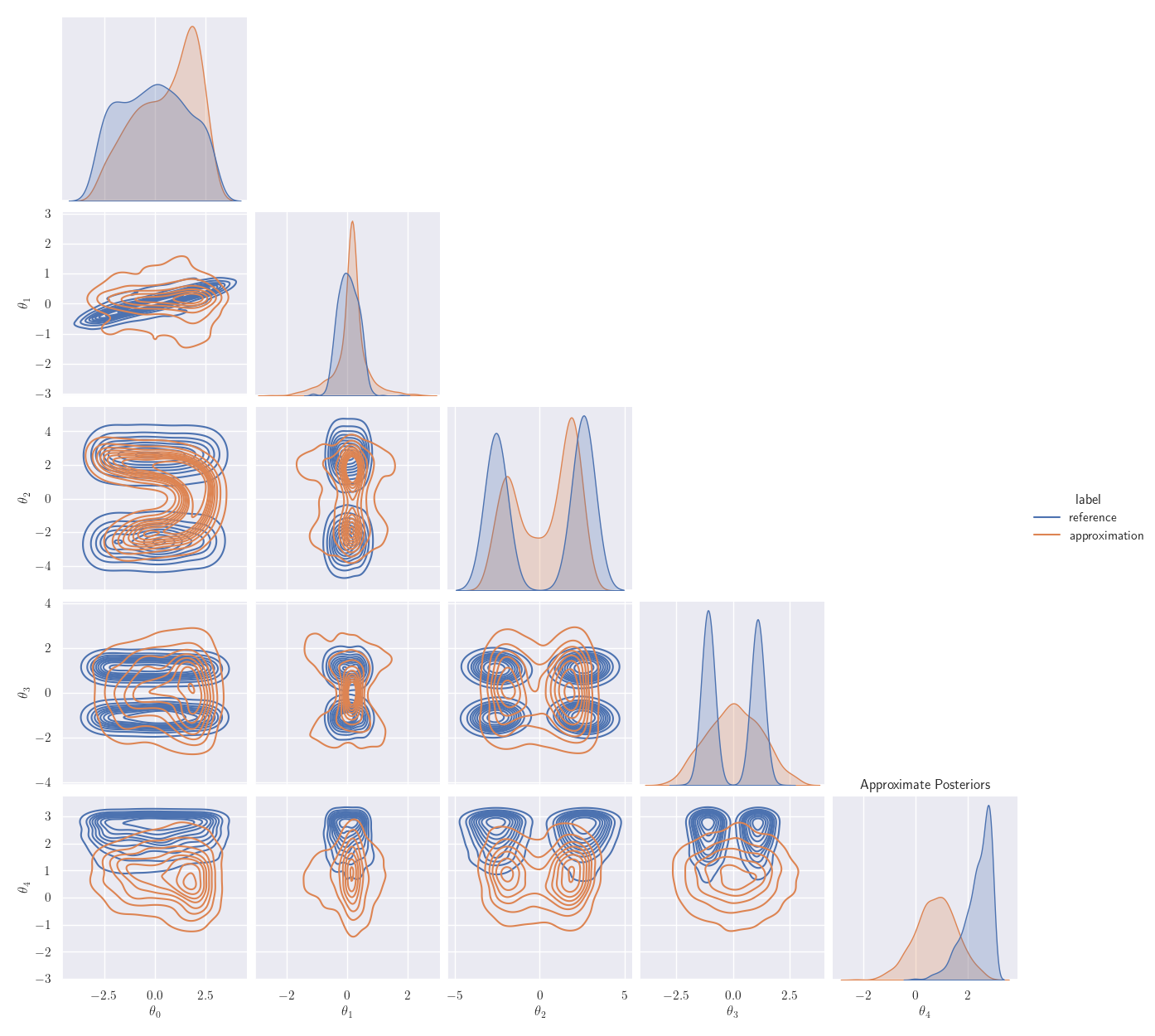}
\centering
\end{figure}


\subsection{Bernoulli GLM}
\begin{figure}[H]
\includegraphics[scale=0.1]{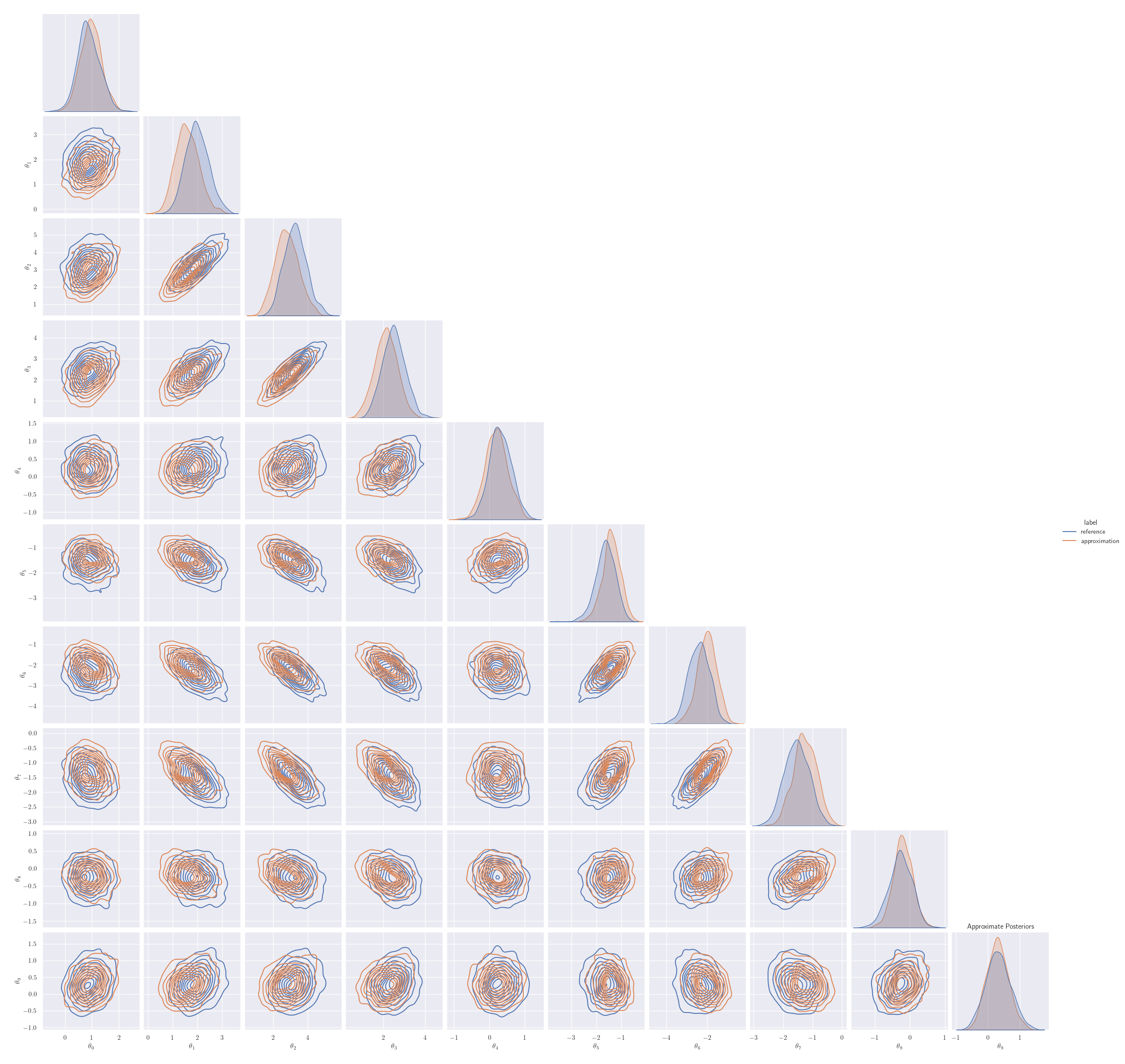}
\centering
\end{figure}





\section{SBI Representative Points}\label{appendix:sbi_rps}

\subsection{Gaussian Mixture}
\begin{figure}[H]
\includegraphics[scale=0.6]{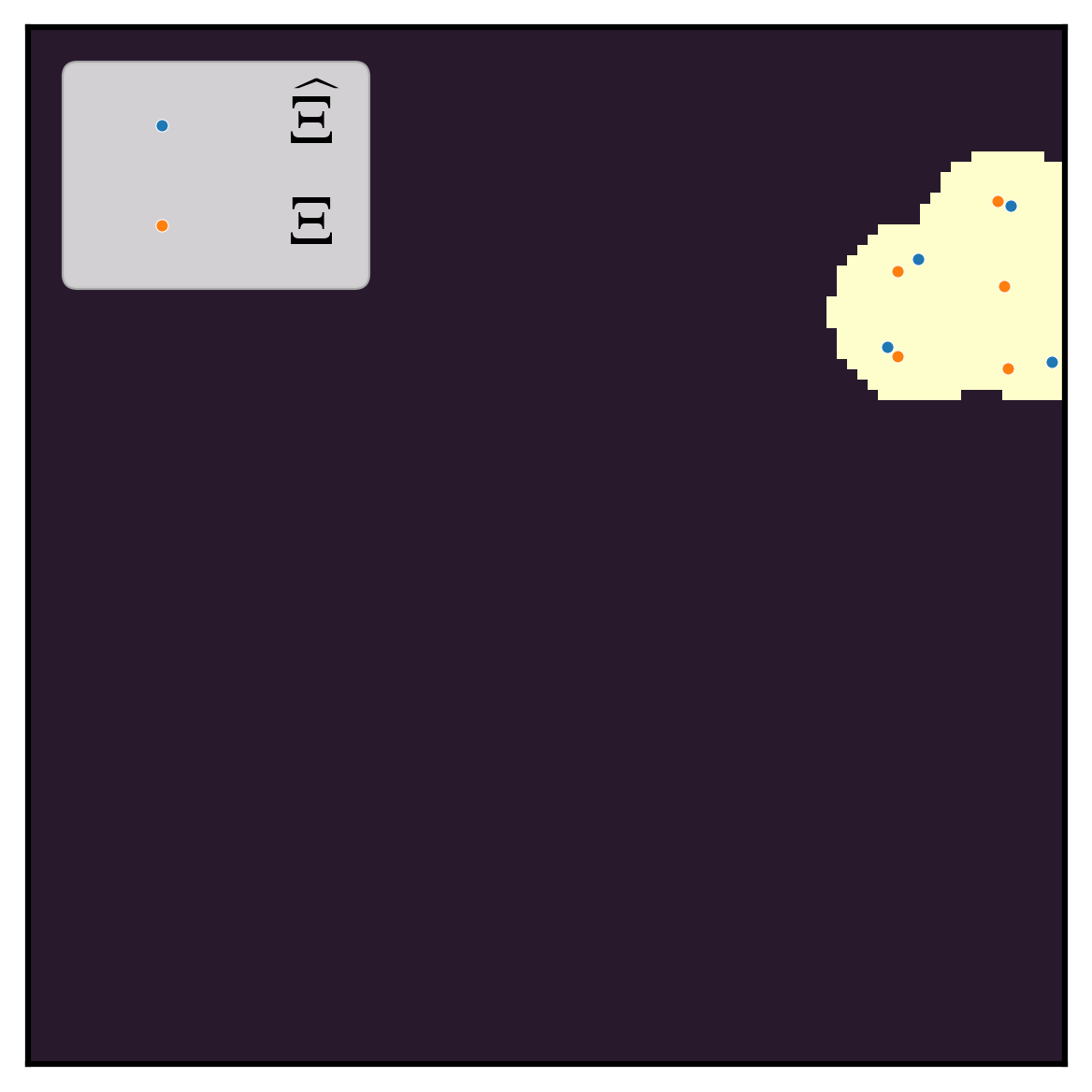}
\centering
\end{figure}

\subsection{Two Moons}
\begin{figure}[H]
\includegraphics[scale=0.6]{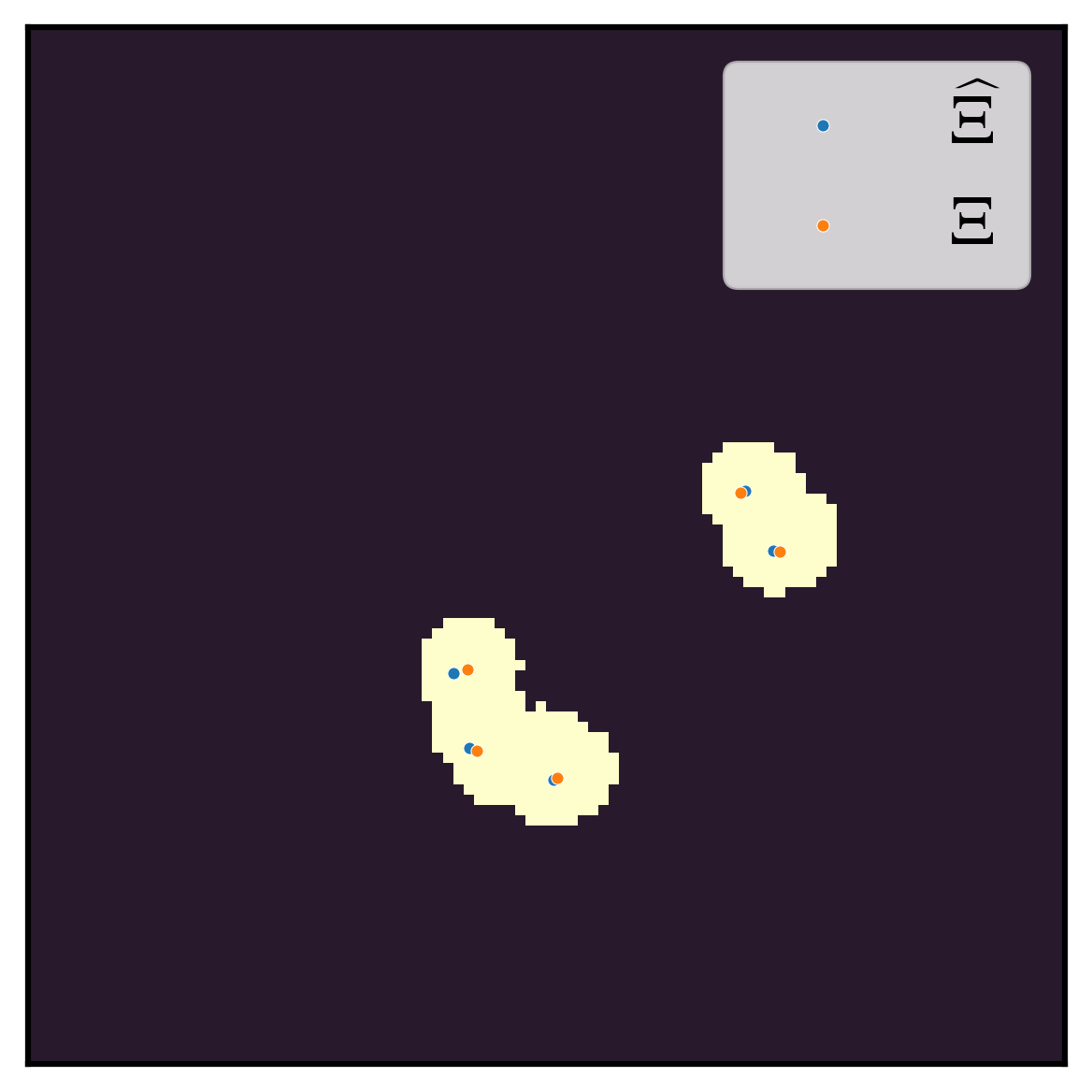}
\centering
\end{figure}



\section{Robust Vehicle Routing Setup}\label{appendix:experiments_routing}
The routing graph of Manhattan was extracted using OSMnx, with local highway speeds extracted using OpenStreetMap \citep{boeing2017osmnx}. Highway speed imputation was performed on edges where such information was not available, specifically by averaging over those highways of comparable categorization, namely ``residential,'' ``secondary,'' or ``tertiary.'' Doing so defined a nominal travel cost $\widetilde{c}$.

We now wish to modify these nominal travel costs to account for the weather predictions made upstream. That is, we wish to account for the precipitation map $\widetilde{Y}\in \mathbb{R}^{W\times H}$ in these edge weights. To do so, we use the global coordinates $(c^{v}_x, c^{v}_y)\in\mathbb{R}^2$ of each $v\in\mathcal{V}$ to find the precipitation at the corresponding location. Concretely, we determine the pixel coordinate by scaling the coordinate to the range of the region that was forecasted. So, for a forecast over the window $(c_x^{\min}, c_x^{\max})\times (c_y^{\min}, c_y^{\max})$, the corresponding pixel lookup is:

$$
p^{v}_x = \floor{\frac{c^{v}_x - c_x^{\min}}{c_x^{\max} - c_x^{\min}}} \times W
\qquad
p^{v}_y = \floor{\frac{c^{v}_y - c_y^{\min}}{c_y^{\max} - c_y^{\min}}} \times H.
$$

\begin{figure*}
\centering
\includegraphics[scale=0.45]{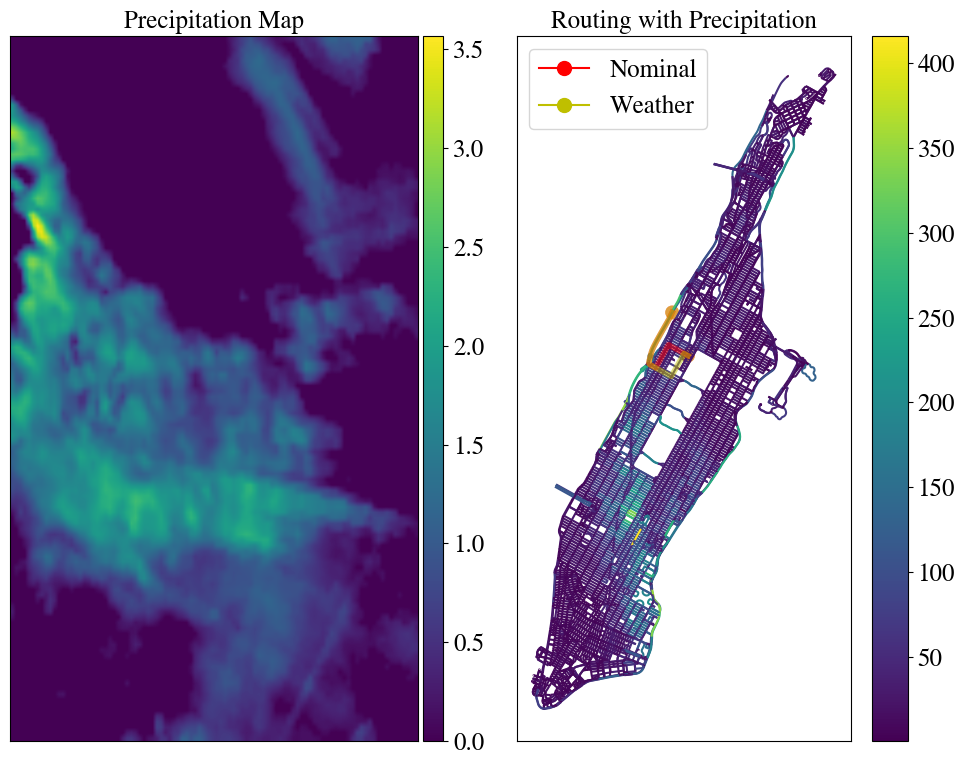}
\caption{\label{fig:example_weather_pred} Precipitation maps (left) are converted to edge weights (right) as per Equation \ref{eqn:forecast_to_weight}. Solving the shortest paths problem (SPP) on this newly weighted graph, therefore, can produce distinct routes from that based on the nominal travel-time SPP, as highlighted by the two distinct paths under the nominal and weather-weighted graphs on the right.}
\end{figure*}

The corresponding precipitation associated with each vertex, therefore, is $\widetilde{Y}_{p^{v}_x, p^{v}_y}$. We define the final travel cost for each edge $e\in\mathcal{E}$ with endpoints $(e_s, e_t)$ as:

\begin{equation}\label{eqn:forecast_to_weight}
    c_{e} := \widetilde{c}_{e} \cdot \exp\left\{\frac{\widetilde{Y}_{p^{e_s}_x, p^{v}_y} + \widetilde{Y}_{p^{e_t}_x, p^{e_t}_y}}{2}\right\}.
\end{equation}

We then solve SPP on the weighted directed graph with edge weights $c_{e}$. An example of this weighting and the corresponding shortest path is illustrated in Figure \ref{fig:example_weather_pred}. 

\end{document}